\documentclass[letterpaper, 11 pt, onecolumn]{ieeeconf}  

\usepackage[utf8]{inputenc}

\usepackage{amscd}
\usepackage{amsfonts}
\usepackage{amsmath}
\usepackage{amssymb}
\usepackage[usenames,dvipsnames]{color}
\usepackage{array}
\usepackage{longtable}
\usepackage{xspace}
\usepackage{paralist}
\usepackage{listings}
\usepackage{extarrows}
\usepackage{xifthen}
\usepackage{url}
  \usepackage{mathpartir}
\usepackage{csquotes}
\usepackage{psfrag}

\newtheorem{theorem}{Theorem}
\newtheorem{example}{Example}
\newtheorem{lemma}{Lemma}
\newtheorem{remark}{Remark}
\newtheorem{definition}{Definition}
\newtheorem{corollary}{Corollary}

\usepackage{pgf}
\usepackage{tikz}
\usetikzlibrary{trees,decorations,arrows,automata,shadows,positioning,plotmarks,backgrounds,shapes}
\usetikzlibrary{calc,matrix,fit,petri,decorations.pathmorphing,patterns}

 \pgfdeclarelayer{foreground}
 \pgfdeclarelayer{background}
 \pgfsetlayers{background,main,foreground}

\newcommand{\REFlem}[1]{\text{Lemma~\ref{#1}}}

\newcommand{\REFthm}[1]{\text{Theorem~\ref{#1}}}

\newcommand{\REFdef}[1]{Definition~\ref{#1}}

\newcommand{\deff}{:=}
\newcommand{\BR}[1]{\left( #1 \right)}

\newcommand{\ON}[1]{\operatorname{#1}}

\def\clap#1{\hbox to 0pt{\hss#1\hss}}

\newcommand{\val}[1]{\ensuremath{\mathsf{#1}}}

\newcommand{\DiCases}[4]{\ensuremath{\begin{cases}%
#1&,~#2\\%
#3&,~#4%
\end{cases}}}%
\makeatletter 
\newif\ifFIRST
\newif\ifSECOND
\let\LISTOP\relax
\newcommand{\List}[4][\;]{#3#1%
	\FIRSTtrue
	\@for\i:=#2\do{%
	\ifFIRST\LISTOP{\i}\FIRSTfalse\else,\LISTOP{\i}\fi%
	}%
	#1#4%
	\let\LISTOP\relax
}
\makeatother
\newcounter{DINGLIST}
\stepcounter{DINGLIST}
\newcommand{\markD}[3][\;\;]{\text{\ding{\the\numexpr171+\theDINGLIST}\stepcounter{DINGLIST}}#1#3}

\newcommand{\ZZ}{\textsl{Zu Zeigen:}~\@ifstar\ZZStar\ZZNoStar}
\newcommand{\ZZStar}[1]{\begin{align*}#1\end{align*}}
\newcommand{\ZZNoStar}[1]{\ensuremath{#1}}
\newcommand{\THATIS}{i.e.\xspace}
\newcommand{\SUCHTHAT}{s.t.\xspace}
\newcommand{\IFF}{\unskip~\text{iff}~}
\newcommand{\SHOW}[2][]{Show \ifthenelse{\isempty{#1}}{}{#1, \THATIS, }\ensuremath{#2}:}

\newcommand{\SORRY}[1]{\emph{\color{red}Sorry: #1}\@latex@warning{SORRY: #1}}

\makeatletter

\newcommand{\propNeg}{\@ifstar\propNegStar\propNegNoStar}
\newcommand{\propNegStar}[1]{\ensuremath{\left(\propNegNoStar{#1}\right)}}
\newcommand{\propNegNoStar}[2][\cdot]{\ensuremath{\neg\ifthenelse{\isempty{#2}}{#1}{#2}}}

\newcommand{\propConj}{\@ifstar\propConjStar\propConjNoStar}
\newcommand{\propConjStar}[2]{\ensuremath{\left(\propConjNoStar{#1}{#2}\right)}}
\newcommand{\propConjNoStar}[3][\cdot]{\ensuremath{\ifthenelse{\isempty{#2}}{#1}{#2}\wedge\ifthenelse{\isempty{#3}}{#1}{#3}}}

\newcommand{\propDisj}{\@ifstar\propDisjStar\propDisjNoStar}
\newcommand{\propDisjStar}[2]{\ensuremath{\left(\propDisjNoStar{#1}{#2}\right)}}
\newcommand{\propDisjNoStar}[3][\cdot]{\ensuremath{\ifthenelse{\isempty{#2}}{#1}{#2}\vee\ifthenelse{\isempty{#3}}{#1}{#3}}}

\newcommand{\propImp}{\@ifstar\propImpStar\propImpNoStar}
\newcommand{\propImpStar}[2]{\ensuremath{\left(\propImpNoStar{#1}{#2}\right)}}
\newcommand{\propImpNoStar}[3][\cdot]{\ensuremath{\ifthenelse{\isempty{#2}}{#1}{#2}\Rightarrow\ifthenelse{\isempty{#3}}{#1}{#3}}}

\newcommand{\propAequ}{\@ifstar\propAequStar\propAequNoStar}
\newcommand{\propAequStar}[2]{\ensuremath{\left(\propAequNoStar{#1}{#2}\right)}}
\newcommand{\propAequNoStar}[3][\cdot]{\ensuremath{\ifthenelse{\isempty{#2}}{#1}{#2}\Leftrightarrow\ifthenelse{\isempty{#3}}{#1}{#3}}}

\newcommand{\propXOR}{\@ifstar\propXORStar\propXORNoStar}
\newcommand{\propXORStar}[2]{\ensuremath{\left(\propXORNoStar{#1}{#2}\right)}}
\newcommand{\propXORNoStar}[3][\cdot]{\ensuremath{\ifthenelse{\isempty{#2}}{#1}{#2}\oplus\ifthenelse{\isempty{#3}}{#1}{#3}}}

\newcommand{\AllQ}{\@ifstar\AllQStar\AllQNoStar}
\newcommand{\AllQStar}[3][\;]{\ensuremath{\left(\forall #2#1.#1#3\right)}}
\newcommand{\AllQNoStar}[3][\;]{\ensuremath{\forall #2#1.#1#3}}
\newcommand{\AllQu}{\@ifstar\AllQuStar\AllQuNoStar}
\newcommand{\AllQuStar}[3][\;]{\ensuremath{\left(\forall^{\infty} #2#1.#1#3\right)}}
\newcommand{\AllQuNoStar}[3][\;]{\ensuremath{\forall^{\infty} #2#1.#1#3}}

\newcommand{\ExQ}{\@ifstar\ExQStar\ExQNoStar}
\newcommand{\ExQStar}[3][\;]{\ensuremath{\left(\exists #2#1.#1#3\right)}}
\newcommand{\ExQNoStar}[3][\;]{\ensuremath{\exists #2#1.#1#3}}

\newcommand{\NExQ}{\@ifstar\NExQStar\NExQNoStar}
\newcommand{\NExQStar}[3][\;]{\ensuremath{\left(\nexists #2#1.#1#3\right)}}
\newcommand{\NExQNoStar}[3][\;]{\ensuremath{\nexists #2#1.#1#3}}

\newcommand{\UniqueQ}{\@ifstar\UniqueQStar\UniqueQNoStar}
\newcommand{\UniqueQStar}[3][\;]{\ensuremath{\left(\exists! #2#1.#1#3\right)}}
\newcommand{\UniqueQNoStar}[3][\;]{\ensuremath{\exists! #2#1.#1#3}}

\newcommand{\Set}[2][]{\List[#1]{#2}{\{}{\}}}
\newcommand{\VSet}[2][]{\let\LISTOP\val\List[#1]{#2}{\{}{\}}}

\newcommand{\Tuple}[2][]{\List[#1]{#2}{(}{)}}
\newcommand{\VTuple}[2][]{\let\LISTOP\val\List[#1]{#2}{(}{)}}

\newcommand{\SetComp}[3][]{\{#1#2#1\mid#1#3#1\}}
\newcommand{\SetCompX}[3][]{\left\{#1#2#1\middle\vert#1#3#1\right\}}

\newcommand{\POWERSET}{\@ifstar\POWERSETStar\POWERSETNoStar}
\newcommand{\POWERSETStar}[1]{\ensuremath{\ON{2}^{\ifthenelse{\isempty{#1}}{\cdot}{#1}}}}
\newcommand{\POWERSETNoStar}[1]{\ensuremath{\ON{2}^{\ifthenelse{\isempty{#1}}{\cdot}{#1}}}}

\newcommand{\FINPOWERSET}{\@ifstar\FINPOWERSETStar\FINPOWERSETNoStar}
\newcommand{\FINPOWERSETStar}[1]{\ensuremath{\mathcal{P}_{\ON{fin}}(\ifthenelse{\isempty{#1}}{\cdot}{#1})}}
\newcommand{\FINPOWERSETNoStar}[1]{\ensuremath{\mathcal{P}_{\ON{fin}}\left(\ifthenelse{\isempty{#1}}{\cdot}{#1}\right)}}

\newcommand{\UNION}{\@ifstar\UNIONStar\UNIONNoStar}
\newcommand{\UNIONStar}[2]{\ensuremath{\left(\UNIONNoStar{#1}{#2}\right)}}
\newcommand{\UNIONNoStar}[2]{\ensuremath{\ifthenelse{\isempty{#1}}{\cdot}{#1}\cup\ifthenelse{\isempty{#2}}{\cdot}{#2}}}

\newcommand{\UNIOND}{\@ifstar\UNIONDStar\UNIONDNoStar}
\newcommand{\UNIONDStar}[2]{\ensuremath{\left(\UNIONDNoStar{#1}{#2}\right)}}
\newcommand{\UNIONDNoStar}[2]{\ensuremath{\ifthenelse{\isempty{#1}}{\cdot}{#1}\uplus\ifthenelse{\isempty{#2}}{\cdot}{#2}}}

\newcommand{\SETMINUS}{\@ifstar\SETMINUSStar\SETMINUSNoStar}
\newcommand{\SETMINUSStar}[2]{\ensuremath{\left(\SETMINUSNoStar{#1}{#2}\right)}}
\newcommand{\SETMINUSNoStar}[2]{\ensuremath{\ifthenelse{\isempty{#1}}{\cdot}{#1}\setminus\ifthenelse{\isempty{#2}}{\cdot}{#2}}}

\newcommand{\INTERSECT}{\@ifstar\INTERSECTStar\INTERSECTNoStar}
\newcommand{\INTERSECTStar}[2]{\ensuremath{\left(\INTERSECTNoStar{#1}{#2}\right)}}
\newcommand{\INTERSECTNoStar}[2]{\ensuremath{\ifthenelse{\isempty{#1}}{\cdot}{#1}\cap\ifthenelse{\isempty{#2}}{\cdot}{#2}}}

\newcommand{\CARTPROD}{\@ifstar\CARTPRODStar\CARTPRODNoStar}
\newcommand{\CARTPRODStar}[2]{\ensuremath{\left(\CARTPRODNoStar{#1}{#2}\right)}}
\newcommand{\CARTPRODNoStar}[2]{\ensuremath{\ifthenelse{\isempty{#1}}{\cdot}{#1}\times\ifthenelse{\isempty{#2}}{\cdot}{#2}}}

\newcommand{\FINCOUNT}{\@ifstar\FinCountStar\FinCountNoStar}
\newcommand{\FinCountStar}[1]{\ensuremath{\#(\ifthenelse{\isempty{#1}}{\cdot}{#1})}}
\newcommand{\FinCountNoStar}[1]{\ensuremath{\#\left(\ifthenelse{\isempty{#1}}{\cdot}{#1}\right)}}

\makeatother 

\newcommand{\parfun}{\ensuremath{\ON{\rightharpoonup}}}
\newcommand{\fun}{\ensuremath{\ON{\rightarrow}}}

\tikzstyle{istate}=[state,initial,initial text=]
\tikzstyle{fistate}=[state,accepting,initial,initial text=]
\tikzstyle{fistateA}=[state,accepting,initial,initial text=,initial where=above]
\tikzstyle{fistateB}=[state,accepting,initial,initial text=,initial where=below]
\tikzstyle{fistateL}=[state,accepting,initial,initial text=,initial where=left]
\tikzstyle{fistateR}=[state,accepting,initial,initial text=,initial where=right]
\tikzstyle{ifstate}=[state,accepting,initial,initial text=]
\tikzstyle{ifstateA}=[state,accepting,initial,initial text=,initial where=above]
\tikzstyle{ifstateB}=[state,accepting,initial,initial text=,initial where=below]
\tikzstyle{ifstateL}=[state,accepting,initial,initial text=,initial where=left]
\tikzstyle{ifstateR}=[state,accepting,initial,initial text=,initial where=right]
\tikzstyle{istateA}=[state,initial,initial text=,initial where=above]
\tikzstyle{istateB}=[state,initial,initial text=,initial where=below]
\tikzstyle{istateL}=[state,initial,initial text=,initial where=left]
\tikzstyle{istateR}=[state,initial,initial text=,initial where=right]
\tikzstyle{fstate}=[state,accepting]
\tikzstyle{SFSautomat}=[->,>=stealth',shorten >=1pt,auto,node distance=2cm,on grid,semithick,inner sep=2pt,bend angle=45]

\newenvironment{propConjA}{\left(\def\unionAtest{1}\begin{array}{@{\if\unionAtest1\gdef\unionAtest{0}\phantom{\wedge}\else\wedge\fi}l@{}}}{\end{array}\right)}

\makeatletter
  \newlength{\SFS@HEIGHT}
  \newlength{\SFS@WIDTH}
  \newcommand{\SplitX}[2]{
	    \settoheight{\SFS@HEIGHT}{$#2$}
	    \settowidth{\SFS@WIDTH}{$#2$}
	    \mbox{\begin{tikzpicture}[baseline=(current bounding box.center)]
	    \node[] (E) at (0,0) {$#1$};
	    \node[inner sep=0pt] (F) at ($(E.south west)+(1ex,-1ex)+(3ex+.5\SFS@WIDTH,-\SFS@HEIGHT)$) {$#2$};
	    \node[] (E) at (0,0) {\phantom{$#1$}};
	    \draw[fill] ($(E.east)+(1ex,0ex)$) circle (.2ex);
	    \draw[-] ($(E.east)+(1ex,0ex)$) -- ($(E.south east)+(1ex,-0.5ex)$) -- ($(E.south west)+(1ex,-0.5ex)$) -- ($(E.south west)+(1ex,-1ex)-(0,\SFS@HEIGHT)$) -- ($(E.south west)+(2.5ex,-1ex)-(0,\SFS@HEIGHT)$);
	    \draw[fill] ($(E.south west)+(2.5ex,-1ex)-(0,\SFS@HEIGHT)$) circle (.2ex);
	    \end{tikzpicture}}}
  \newcommand{\SplitS}[2]{
	    \settoheight{\SFS@HEIGHT}{$#2$}
	    \settowidth{\SFS@WIDTH}{$#2$}
	    \mbox{\begin{tikzpicture}[baseline=(current bounding box.center)]
	    \node[] (E) at (0,0) {$#1$};
	    \node[inner sep=0pt] (F) at ($(E.south west)+(1ex,0.5ex)+(3ex+.5\SFS@WIDTH,-\SFS@HEIGHT)$) {$#2$};
	    \end{tikzpicture}}}	    
  \newcommand{\AllQSplit}[2]{\SplitX{\forall\;#1\;.}{#2}}

  \newcommand{\ExQSplit}[2]{\SplitX{\exists\;#1\;.}{#2}}

  \newcommand{\propImpSplit}[2]{\SplitX{#1\;\Rightarrow\;}{#2}}

\makeatother

\providecommand{\length}[1]{\lvert#1\rvert}

\newcommand{\trivialN}[1]{\text{trivial}\xspace}

\newcommand{\Rb}{\ensuremath{\mathbb{R}}} 
 
\newcommand{\Nbn}{\ensuremath{\mathbb{N}_{0}}}
\newcommand{\Rbn}{\ensuremath{\mathbb{R}_{0}^+}}

\newcommand{\twoup}[1]{\ensuremath{2^{#1}}} 

\newcommand{\X}{\ensuremath{X}} 
\newcommand{\Xt}[2]{\ensuremath{\ifthenelse{\isempty{#2}}{X_{I}^{#1}}{X_{I,#2}^{#1}}}} 
\newcommand{\XT}[1]{\ensuremath{\ifthenelse{\isempty{#1}}{X_I}{X_{I,#1}}}} 
\newcommand{\Xk}[2]{\ensuremath{\ifthenelse{\isempty{#2}}{X_{E}^{#1}}{X_{E,#2}^{#1}}}} 
\newcommand{\XK}[1]{\ensuremath{\ifthenelse{\isempty{#1}}{X_{E}}{X_{E,#1}}}} 

\newcommand{\Zk}[2]{\ensuremath{\ifthenelse{\isempty{#2}}{Z_{\TE,#1}}{Z_{#2,\TE,#1}}}}

\newcommand{\x}{\ensuremath{x}}

\newcommand{\Zt}[1]{\ensuremath{\ifthenelse{\isempty{#1}}{Z_t}{Z_{#1,t}}}} 
\newcommand{\ZT}[1]{\ensuremath{\ifthenelse{\isempty{#1}}{Z_T}{Z_{#1,T}}}} 
 
\newcommand{\ZPit}[1]{\ensuremath{\ifthenelse{\isempty{#1}}{Z_t}{Z_{#1,t}}}} 
\newcommand{\ZPiT}[1]{\ensuremath{\ifthenelse{\isempty{#1}}{Z_T}{Z_{#1,T}}}} 
 
\newcommand{\ZTit}[1]{\ensuremath{\ifthenelse{\isempty{#1}}{\breve{Z}_t}{\breve{Z}_{#1,t}}}} 
\newcommand{\ZTiT}[1]{\ensuremath{\ifthenelse{\isempty{#1}}{\breve{Z}_T}{\breve{Z}_{#1,T}}}} 
\newcommand{\z}{\ensuremath{z}}

\newcommand{\U}{\ensuremath{U}}

\newcommand{\w}{\ensuremath{w}}

\newcommand{\Y}{\ensuremath{Y}} 

\newcommand{\T}{\ensuremath{T}} 
\newcommand{\I}{\ensuremath{\mathbf{i}}}

\renewcommand{\ll}[1]{\ensuremath{|_{[#1]}}}

\newcommand{\f}[1]{\ensuremath{f\ifthenelse{\isempty{#1}}{}{\Tuple{#1}}}} 
\newcommand{\g}[1]{\ensuremath{g\ifthenelse{\isempty{#1}}{}{\Tuple{#1}}}} 
\newcommand{\h}[1]{\ensuremath{h\ifthenelse{\isempty{#1}}{}{\Tuple{#1}}}} 
\newcommand{\fs}[2]{\ensuremath{f_{#2}\ifthenelse{\isempty{#1}}{}{\Tuple{#1}}}} 
\newcommand{\gs}[2]{\ensuremath{g_{#2}\ifthenelse{\isempty{#1}}{}{\Tuple{#1}}}}  
\newcommand{\hs}[2]{\ensuremath{h_{#2}\ifthenelse{\isempty{#1}}{}{\Tuple{#1}}}}  

\newcommand{\lnc}{\ensuremath{l_t}}

\newcommand{\Beh}{\ensuremath{\mathcal{B}}}

\newcommand{\kg}[1]{\ensuremath{\xspace\preceq_{#1}\xspace}}
\newcommand{\hg}[1]{\ensuremath{\xspace\cong_{#1}\xspace}}
\newcommand{\async}{\ensuremath{\wr_{|}}}
\newcommand{\sync}{\ensuremath{\shortparallel}}
\newcommand{\wsync}{\ensuremath{\wr\shortmid}}
\newcommand{\SR}[3]{\ensuremath{\mathfrak{R}_{#1}(#2,#3)}}

\newcommand{\ESn}[1]{\ensuremath{\ifthenelse{\isempty{#1}}{\Sigma^+_{S}}{\Sigma^+_{S,#1}}}}

\newcommand{\BehS}[1]{\ensuremath{\ifthenelse{\isempty{#1}}{\Beh_{S}}{\Beh_{S,#1}}}}
\newcommand{\BehE}[1]{\ensuremath{\ifthenelse{\isempty{#1}}{\Beh_{E}}{\Beh_{E,#1}}}}

\newcommand{\WT}{\ensuremath{W}}
\newcommand{\W}{\ensuremath{W}}
\newcommand{\D}{\ensuremath{D}}
\newcommand{\V}{\ensuremath{V}}
\newcommand{\statemap}[3]{\ifthenelse{\isempty{#2#3}}{\psi_{#1}}{\psi_{#1}(#2,#3)}}
\newcommand{\statemapPi}[3]{\ifthenelse{\isempty{#2#3}}{\varphi_{#1}}{\varphi_{#1}(#2,#3)}}
\newcommand{\statemapTi}[3]{\ifthenelse{\isempty{#2#3}}{\psi_{#1}}{\psi_{#1}(#2,#3)}}

\newcommand{\CONCAT}[4]{#1\wedge^{#2}_{#3}#4}

\newcommand{\Xx}[2]{\ensuremath{\chi_{#1}\ifthenelse{\isempty{#2}}{}{(#2)}}}
\newcommand{\Xxp}[2]{\ensuremath{\overline{\chi}_{#1}\ifthenelse{\isempty{#2}}{}{(#2)}}}
\newcommand{\Xxr}[3]{\ensuremath{\chi_{#1}^{#2}\ifthenelse{\isempty{#3}}{}{(#3)}}}
\newcommand{\Xxrp}[3]{\ensuremath{\overline{\chi}_{#1}^{#2}\ifthenelse{\isempty{#3}}{}{(#3)}}}

\newcommand{\R}{\ensuremath{\mathcal{R}}}

\newcommand{\timescale}{\mathcal{T}}

\newcommand{\timescaleUp}[1]{{#1}}
\newcommand{\timescaleDown}[1]{{#1}^{-1}}
\newcommand{\signalmap}{\phi}
\newcommand{\signalmapR}[1]{\ensuremath{\ifthenelse{\isempty{#1}}{\signalmap_{r}}{\signalmap_{r,#1}}}}
\newcommand{\signalmapP}[1]{\ensuremath{\ifthenelse{\isempty{#1}}{\signalmap_{p}}{\signalmap_{p,#1}}}}
\newcommand{\signalmapT}[1]{\ensuremath{\ifthenelse{\isempty{#1}}{\signalmap_{t}}{\signalmap_{t,#1}}}}
\newcommand{\TE}{\ensuremath{{T_E}}}

\newcommand{\E}{\ensuremath{\Sigma}}
\newcommand{\Ep}[1]{\ensuremath{\Sigma_{#1}^{\signalmap}}}
\newcommand{\EpRhs}{\ensuremath{\Tuple{\T,\TE,\WT,\Gamma,\Beh,\BehE{},\signalmap}}}

\newcommand{\ES}[1]{\ensuremath{\ifthenelse{\isempty{#1}}{\Sigma_{S}}{\Sigma_{S,#1}}}}
\newcommand{\EpS}[1]{\ensuremath{\ifthenelse{\isempty{#1}}{\Ep{S}}{\Ep{S,#1}}}}
\newcommand{\EpSR}[1]{\ensuremath{\ifthenelse{\isempty{#1}}{\Sigma^{\signalmap_{r}}_{S}}{\Sigma^{\signalmap_{r,#1}}_{S,#1}}}}
\newcommand{\EpSP}[1]{\ensuremath{\ifthenelse{\isempty{#1}}{\Sigma^{\signalmap_{p}}_{S}}{\Sigma^{\signalmap_{p,#1}}_{S,#1}}}}
\newcommand{\EpST}[1]{\ensuremath{\ifthenelse{\isempty{#1}}{\Sigma^{\signalmap_{t}}_{S}}{\Sigma^{\signalmap_{t,#1}}_{S,#1}}}}
\newcommand{\EpSRhs}[1]{\ensuremath{(\T_{#1},\allowbreak\TE,\allowbreak\WT_{#1}\nobreak\times\nobreak\X_{#1},\allowbreak\Gamma,\allowbreak\BehS{#1},\allowbreak\BehE{#1},\allowbreak\signalmap_{#1})}}

\newcommand{\EE}[1]{\ensuremath{\ifthenelse{\isempty{#1}}{\Sigma_{E}}{\Sigma_{E,#1}}}}
\newcommand{\ESm}[1]{\ensuremath{\ifthenelse{\isempty{#1}}{\Sigma_{\psi}}{\Sigma_{\psi,#1}}}}

\newcommand{\EplMaxS}[1]{\ensuremath{\ifthenelse{\isempty{#1}}{\Sigma^{\phi,l^\uparrow}_{S}}{\Sigma^{\phi_{#1},l^\uparrow}_{S,#1}}}}
\newcommand{\EplMaxSp}[1]{\ensuremath{\ifthenelse{\isempty{#1}}{\overline{\Sigma}^{\phi,l^\uparrow}_{S}}{\overline{\Sigma}^{\phi_{#1},l^\uparrow}_{S,#1}}}}
\newcommand{\EplncMaxS}[1]{\ensuremath{\ifthenelse{\isempty{#1}}{\Sigma^{\phi,\lnc^\uparrow}_{S}}{\Sigma^{\phi_{#1},\lnc^\uparrow}_{S,#1}}}}
\newcommand{\EplsMaxS}[2]{\ensuremath{\ifthenelse{\isempty{#1}}{\Sigma^{\phi,{#2}^\uparrow}_{S}}{\Sigma^{\phi_{#1},{#2}^\uparrow}_{S,#1}}}}

\newcommand{\ElMaxS}[1]{\ensuremath{\ifthenelse{\isempty{#1}}{\Sigma^{l^\uparrow}_{S}}{\Sigma^{l^\uparrow}_{S,#1}}}}
\newcommand{\ElMaxSp}[1]{\ensuremath{\ifthenelse{\isempty{#1}}{\overline{\Sigma}^{l^\uparrow}_{S}}{\overline{\Sigma}^{l^\uparrow}_{S,#1}}}}
\newcommand{\ElncMaxS}[1]{\ensuremath{\ifthenelse{\isempty{#1}}{\Sigma^{\lnc^\uparrow}_{S}}{\Sigma^{\lnc^\uparrow}_{S,#1}}}}
\newcommand{\ElsMaxS}[2]{\ensuremath{\ifthenelse{\isempty{#1}}{\Sigma^{{#2}^\uparrow}_{S}}{\Sigma^{{#2}^\uparrow}_{S,#1}}}}

\newcommand{\ElE}[1]{\ensuremath{\ifthenelse{\isempty{#1}}{\Sigma^{l}_{E}}{\Sigma^{l}_{E,#1}}}}
\newcommand{\ElEp}[1]{\ensuremath{\ifthenelse{\isempty{#1}}{\Sigma^{l}_{E}}{\overline{\Sigma}^{l}_{E,#1}}}}
\newcommand{\ElncE}[1]{\ensuremath{\ifthenelse{\isempty{#1}}{\Sigma^{\lnc}_{E}}{\Sigma^{\lnc}_{E,#1}}}}
\newcommand{\ElMaxE}[1]{\ensuremath{\ifthenelse{\isempty{#1}}{\Sigma^{l^\uparrow}_{E}}{\Sigma^{l^\uparrow}_{E,#1}}}}
\newcommand{\ElMaxEp}[1]{\ensuremath{\ifthenelse{\isempty{#1}}{\overline{\Sigma}^{l^\uparrow}_{E}}{\overline{\Sigma}^{l^\uparrow}_{E,#1}}}}
\newcommand{\ElncMaxE}[1]{\ensuremath{\ifthenelse{\isempty{#1}}{\Sigma^{\lnc^\uparrow}_{E}}{\Sigma^{\lnc^\uparrow}_{E,#1}}}}
\newcommand{\ElsMaxE}[2]{\ensuremath{\ifthenelse{\isempty{#1}}{\Sigma^{{#2}^\uparrow}_{E}}{\Sigma^{{#2}^\uparrow}_{E,#1}}}}

\newcommand{\Behp}{\ensuremath{\Beh^{\signalmap}}}
\newcommand{\BehpS}[1]{\ensuremath{\ifthenelse{\isempty{#1}}{\Behp_{S}}{\Beh^{\signalmap_{#1}}_{S,#1}}}}

\newcommand{\BehlMaxS}[1]{\ensuremath{\ifthenelse{\isempty{#1}}{\Beh^{l^\uparrow}_{S}}{\Beh^{l^\uparrow}_{S,#1}}}}
\newcommand{\BehlMaxSp}[1]{\ensuremath{\ifthenelse{\isempty{#1}}{\overline{\Beh}^{l^\uparrow}_{S}}{\overline{\Beh}^{l^\uparrow}_{S,#1}}}}
\newcommand{\BehlncMaxS}[1]{\ensuremath{\ifthenelse{\isempty{#1}}{\Beh^{\lnc^\uparrow}_{S}}{\Beh^{\lnc^\uparrow}_{S,#1}}}}
\newcommand{\BehlsMaxS}[2]{\ensuremath{\ifthenelse{\isempty{#1}}{\Beh^{{#2}^\uparrow}_{S}}{\Beh^{{#2}^\uparrow}_{S,#1}}}}

\newcommand{\BehlE}[1]{\ensuremath{\ifthenelse{\isempty{#1}}{\Beh^{l}_{E}}{\Beh^{l}_{E,#1}}}}
\newcommand{\BehlncE}[1]{\ensuremath{\ifthenelse{\isempty{#1}}{\Beh^{\lnc}_{E}}{\Beh^{\lnc}_{E,#1}}}}
\newcommand{\BehlMaxE}[1]{\ensuremath{\ifthenelse{\isempty{#1}}{\Beh^{l^\uparrow}_{E}}{\Beh^{l^\uparrow}_{E,#1}}}}
\newcommand{\BehlMaxEp}[1]{\ensuremath{\ifthenelse{\isempty{#1}}{\overline{\Beh}^{l^\uparrow}_{E}}{\overline{\Beh}^{l^\uparrow}_{E,#1}}}}
\newcommand{\BehlncMaxE}[1]{\ensuremath{\ifthenelse{\isempty{#1}}{\Beh^{\lnc^\uparrow}_{E}}{\Beh^{\lnc^\uparrow}_{E,#1}}}}
\newcommand{\BehlsMaxE}[2]{\ensuremath{\ifthenelse{\isempty{#1}}{\Beh^{{#2}^\uparrow}_{E}}{\Beh^{{#2}^\uparrow}_{E,#1}}}}

\newcommand{\BehVlMaxS}[1]{\ensuremath{\ifthenelse{\isempty{#1}}{\projState{\V}{\Beh^{l^\uparrow}_{S}}}{\projState{\V}{\Beh^{l^\uparrow}_{S,#1}}}}}

\newcommand{\BehDlMaxS}[1]{\ensuremath{\ifthenelse{\isempty{#1}}{\projState{\D}{\Beh^{l^\uparrow}_{S}}}{\projState{\D}{\Beh^{l^\uparrow}_{S,#1}}}}}

\newcommand{\EoMaxS}[1]{\ensuremath{\ifthenelse{\isempty{#1}}{\Sigma^{1^\uparrow}_{S}}{\Sigma^{1^\uparrow}_{S,#1}}}}

\newcommand{\projState}[2]{\pi_{#1}(#2)}

\newcommand{\dom}[1]{\ensuremath{\mathrm{dom(#1)}}}

\title{\LARGE \bf
Simulation and Bisimulation over Multiple Time Scales\\ in a Behavioral Setting
}

\author{Anne-Kathrin Schmuck and Jörg Raisch
\thanks{A.-K. Schmuck and J. Raisch are with the Control Systems Group, Technical University of Berlin, Germany. J. Raisch is also with the Max Planck Institute
for Dynamics of Complex Technical Systems, Magdeburg, Germany. {\tt\small \{a.schmuck,raisch\}@control.tu-berlin.de}}
}

\begin{document}

\maketitle
\thispagestyle{empty}
\pagestyle{empty}

 \begin{abstract}
\noindent This paper introduces a new behavioral system model with distinct external and internal signals possibly evolving on different time scales. This allows to capture abstraction processes or signal aggregation in the context of control and verification of large scale systems. For this new system model different notions of simulation and bisimulation are derived, ensuring that they are, respectively, preorders and equivalence relations for the system class under consideration.\\
These relations can capture a wide selection of similarity notions available in the literature. This paper therefore provides a suitable framework for their comparison.\\
\end{abstract}

\section{Introduction}
State explosion is a very common problem in the control of large scale systems due to the interconnection of numerous subsystems. Therefore, it is usually desired to reduce the state space of subsystems while overapproximating or preserving their external behavior important for their interconnection to surrounding components.\\
This mechanism is also used to reduce the complexity of verification problems in the theoretical computer science community. Here, systems are usually modeled by so called transition systems, a subclass of discrete time state space models. For these models, the notion of bisimilarity plays an important role. This concept was introduced by Milner \cite{Milner1989} in the context of concurrent processes to describe how state trajectories of two transition systems mimic each other while producing the same ``external'' behavior, i.e., using the same transition symbols. If such a bisimulation relation exists, it was shown that many interesting properties expressible in temporal logics, in particular reachability, are preserved when replacing a system by a bisimilar one.\\
The use of bisimulation relations for other system models was discussed in the survey paper \cite{AlurHenzingerLaffarrierePappas2000}. Here, special classes of hybrid systems are rewritten into a transition system and it was shown that they allow for purely discrete abstractions bisimilar to the constructed transition system. Pappas \cite{Pappas2003} adapted this method for linear time-invariant continuous state space models with finite observation maps, still using both a rewriting and an abstraction step.
To remove the rewriting step, van der Schaft \cite{Schaft2004} introduced a notion of bisimulation directly applicable to continuous systems. He showed that this equivalence interpretation unifies the concepts of state space equivalence and reduction using controlled invariant subspaces. 
These results where generalized by van der Schaft and coworkers to hybrid systems \cite{Schaft2004Hyb}, switched linear systems \cite{PolaSchaft2006} and behavioral systems 
\cite{JuliusSchaft2005}.\\
Recently, Davoren and Tabuada \cite{DavorenTabuada2007} presented simulation 
and bisimulation relations using general flow systems \cite{Davoren2004}, preserving properties formulated in the so called general flow logic \cite{Davoren2004}. General flow systems are able to model continuous, discrete, hybrid or even "meta-hybrid" autonomous state dynamics also allowing equivalence relations between systems with different time scales. This feature extends all previous approaches where only relations between systems with unique time scales are possible. Although Davoren and Moor discussed in \cite{DavorenMoor2006} how general flow systems can be equipped with input and output maps,
the simulation relations in \cite{DavorenTabuada2007} do not incorporate the feature of ensuring identical external signals of bisimilar systems. In \cite{CuijpersReniers2008} a comparison between simulation relations on transition systems and  simulation relations  on general flow systems is presented.\\
Tabuada and coworkers extended the work of Alur et.al. \cite{AlurHenzingerLaffarrierePappas2000} towards finite state abstraction methods ensuring similarity or bisimilarity between the original and the abstracted system \cite{TabuadaPappas2003b,TabuadaPappas2003,Tabuada2006b,Tabuada2008,GirardPolaTabuada2010,TabuadaBook}. 
Independently from this work, the notion of $l$-complete abstraction \cite{MoorRaisch1999} evolved as a discrete abstraction technique in the framework of behavioral systems theory \cite{Willems1989}. %
In both frameworks a finite state abstraction of a possibly continuous or hybrid dynamical system is obtained if the external signal space is finite and the trajectories of external signals evolve on the discrete time axis $\Nbn$.
In the context of bisimilarity relations, these external signals should be preserved during abstraction.
This raises the problem of deriving a bisimilarity notion that ensures equivalence of discrete external signals while comparing state trajectories that evolve on possibly continuous or hybrid time lines. 
This issue has up until now not been explicitly addressed, neither in the context of $l$-complete approximations nor in the work by Tabuada and coworkers. In the latter, as in \cite{AlurHenzingerLaffarrierePappas2000} and \cite{Pappas2003}, the original system is first rewritten into a transition system, previous to the abstraction step. The bisimulation relation is then only ensured to hold between the transition system and its abstraction.\\
To also incorporate the rewriting step into the exploration of equivalence, we introduce a system model with distinct external and internal signals possibly evolving on a different time axis in Section \ref{sec:sys}. To cover a very general class of systems, we use behavioral systems theory \cite{Willems1989} to formalize our notion. We note that this restricts each time axis to be either continuous or discrete. It is future research to also incorporate hybrid time scales for the internal signals as formalized, for example, in \cite{DavorenMoor2006}.
Inspired by the the work in \cite{JuliusSchaft2005,JuliusPhdThesis2005} and \cite{DavorenTabuada2007}, we derive a simulation relation for the newly introduced system model in Section \ref{sec:SimRel}. We show that the introduced simulation and bisimulation relations are preorders and equivalence relations, respectively, for the system class under consideration.\\ 
This work is a first step towards the comparison of different existing approaches to construct (bi)similar finite state abstractions. Due to page limitations this comparison is only shortly touched in various remarks and will be explored in more detail in subsequent publications.\\

\section{Preliminaries}\label{sec:prelim}

A \textit{dynamical system} is given by $\E=\Tuple{\T,\WT,\Beh}$, consisting of the right-unbounded time axis $\T\subseteq\Rb$, the signal space $\WT$ and the behavior of the system $\Beh\subseteq\WT^\T$, where
$\WT^\T\deff\SetComp{w}{w:\T\fun\WT}$ is the set of all 
\textit{signals} evolving on $\T$ and
taking values in $\WT$. 
Slightly abusing notation, we also write  $v\in\WT^\T$ if $v:\T\parfun\WT$ is a \textit{partial function}. This is understood to be shorthand for $v\in\WT^{\dom{v}}$, where $\dom{v}=\SetComp{t\in\T}{v(t)\text{ is defined}}$ is the \textit{domain} of $v$. 
Furthermore, $\I:\T\fun \T$ is the \textit{identity map} s.t.\footnote{Throughout this paper we use the notation "$\AllQ{}{}$", meaning that all statements after the dot hold for all variables in front of the dot. "$\ExQ{}{}$" is interpreted analogously. } $\AllQ{t\in\T}{\I(t)=t}$.
Now let $\WT=\WT_1\times\WT_2$ be a product space. Then the \textit{projection} of a signal $w\in\WT^\T$ to $\WT_1^\T$ is given by $\projState{\WT_1}{w}\deff\SetComp{w_1\in\WT_1^\T}{\ExQ{w_2\in\WT_2^\T}{w=\Tuple{w_1,w_2}}}$ and $\projState{\WT_1}{\Beh}$ denotes the projection of all trajectories in the behavior.  
Given two signals $w_1,w_2\in\WT^\T$ and two points in time $t_1,t_2\in\T$, the \textit{concatenation}  $w_3=\CONCAT{w_1}{t_1}{t_2}{w_2}$ is given by
\begin{equation}\label{equ:concat}
\AllQ{t\in\T}{w_3(t)=\DiCases{w_1(t)}{t< t_1}{w_2(t-t_1+t_2)}{t\geq t_1}},
\end{equation}
where we denote $\CONCAT{\cdot}{t}{t}{\cdot}$ by $\CONCAT{\cdot}{}{t}{\cdot}$.\\

\section{$\signalmap$~-~Dynamical Systems}\label{sec:sys}

When reasoning about similarity and bisimilarity of systems one has to distinguish between  ``external'' signals, which are required to match or satisfy an inclusion property, and the remaining ``internal'' signals. Depending on the chosen system representation and/or the real world problem at hand, this distinction may differ. To incorporate a wide range of possibilities, we define a so called $\signalmap$-dynamical system, where $\signalmap$ is a set-valued map which describes the relation between internal and external signals.\\

\begin{definition}\label{def:DynSysInducesFiniteV}
Let $\E=\Tuple{\T,\WT,\Beh}$ be a dynamical system.
Then $\Ep{}=\EpRhs$ is a \textbf{$\signalmap$-dynamical system} if
\begin{equation*}
 \signalmap:\Beh\fun\twoup{\Gamma^\TE\times\timescale}
\end{equation*}
where $\Gamma$ is an external signal space, $\TE\subseteq\T$ is a right-unbounded time axis, 
\[\timescale=\SetCompX{\timescaleUp{\tau}:\T\parfun\TE}{
 \SplitS{\timescaleUp{\tau} \text{ is surjective and}}{\text{monotonically increasing}}}\] 
 is a set of time scale transformations and 
 \begin{equation}\label{equ:BehE}
  \BehE{}=\SetCompX{\gamma\in\Gamma^\TE}{
 \ExQ{w\in\Beh,\tau\in\timescale}{\Tuple{\gamma,\tau}\in\signalmap(\w)}}
 \end{equation}
 is the external behavior. Furthermore, $\timescaleDown{\tau}:\TE\fun\twoup{\T}$ denotes the inverse time scale transformation\footnote{If $\AllQ{k\in\TE}{\length{\timescaleDown{\tau}(k)}=1}$, by slightly abusing notation, we denote the unique element $t_k\in\timescaleDown{\tau}(k)$ by $\timescaleDown{\tau}(k)$ itself and write $t_k=\timescaleDown{\tau}(k)$.}, i.e., $\timescaleDown{\tau}(k)=\SetComp{t\in\T}{\timescaleUp{\tau}(t)=k}$.
\end{definition}

\begin{remark}
The construction of $\signalmap$ in \REFdef{def:DynSysInducesFiniteV} was inspired by the deterministic map in \cite[Def. 12]{MoorRaischDavoren2003}. Note, that the map in \cite[Def. 12]{MoorRaischDavoren2003} is required to be strictly causal. 
In analogy, one would typically require that the map $\signalmap$ is non-anticipating, i.e., 
 \begin{equation*}\label{equ:rem:causal}
 \AllQSplit{w,w'\in\Beh,\gamma,\gamma'\in\Gamma^\TE,\tau,\tau'\in\timescale,t\in\T}{
\propImp{
\begin{propConjA}
\Tuple{\gamma,\tau}\in\signalmap(\w)\\
\Tuple{\gamma',\tau'}\in\signalmap(\w')\\
\w\ll{0,t}=\w'\ll{0,t}
\end{propConjA}
}{
\ExQ{\tilde{\gamma}\in\Gamma^\TE, \tilde{\tau}\in\timescale}{
\begin{propConjA}
\Tuple{\tilde{\gamma},\tilde{\tau}}\in\signalmap(\w')\\
 \tau\ll{0,t}=\tilde{\tau}\ll{0,t}\\
\gamma\ll{0,\timescaleUp{\tau}(t)}=\tilde{\gamma}\ll{0,\timescaleUp{\tau}'(t)}
\end{propConjA}.
}
}}
\end{equation*}
In words: if we change the future of $w$, the past and present of both $\gamma$ and $\tau$ are allowed to remain unaffected.
\end{remark}

Using this concept, systems with single time axis, i.e., $\T=\TE$, as well as systems with multiple time axes, i.e., $\T\neq\TE$ can be described in a unified fashion.\\
%
 As outlined in the introduction, a large portion of research on simulation relations in the control systems community uses a single time scale. In this context, the signals that are externally visible ``live'' in a subspace of the signal space $\W$. Capturing these models in our framework leads to an identity time scale transformation and a signal map $\signalmap$ projecting signals $w\in\WT^\T$ to the externally visible subspace $\Gamma$.\\
\begin{remark}\label{rem:0}
  Consider a dynamical system ${\E=\Tuple{\T,\WT,\Beh}}$ with $\T=\Nbn$ and $\W=\U\times\Y$, where $\U$ is the set of inputs and $\Y$ is the set of outputs. With a special choice of $\Beh$, this model can capture the dynamics of a transition system as used by Pappas and Tabuada, e.g., in \cite{Pappas2003,TabuadaPappas2003b}. There it is assumed that the inputs are chosen and only the output signals are required to be (bi)simulated by a related system. This can be expressed by a $\signalmap$-dynamical system by choosing $\TE=\Nbn$,  ${\Gamma=\Y}$ and $\AllQ{{\Tuple{u,y}\in\W^\T}}{\signalmap(\Tuple{u,y})=\Set{\Tuple{y,\I}}}$.\\
 Analogously, using $\T=\Rbn$ and $\W=\U\times\Y\times\D$, where $\D$ is the disturbance space, we can construct $\Beh$ such that $\E$ captures the dynamics of the linear time invariant system used by van der Schaft in \cite{Schaft2004}. There, the inputs and outputs are required to match for bisimilar systems. This can be expressed by a $\signalmap$-dynamical system by choosing $\TE=\Rbn$,  ${\Gamma=\U\times\Y}$ and $\AllQ{{\Tuple{u,y,d}\in\W^\T}}{\signalmap(\Tuple{u,y,d})=\Set{\Tuple{\Tuple{u,y},\I}}}$.
\end{remark}

In contrast to the cases described in Remark~\ref{rem:0}, the construction of a $\signalmap$-dynamical system with $\T\neq\TE$ is not as straightforward and therefore illustrated by an example.\\

\begin{example}\label{exp:1}
 Consider a dynamical system ${\E=\Tuple{\T,\WT,\Beh}}$ with $\T=\Rbn$, $\WT=\INTERSECT{\Rb}{[0,40]}$ and $w\in\Beh$ iff $w$ is continuouse.
 Using $\TE=\Nbn$, $\Gamma=\Set{q_1,q_2,q_3,q_4}$ and the sets
 \begin{align*}
  &I_{q_1}=[0,11),&I_{q_2}=(9,21),\\
  &I_{q_3}=(19,31),&I_{q_4}=(29,40],
 \end{align*}
 the external signals $\gamma\in\BehE{}$ are constructed 
 via the discretization $\mathfrak{d}:W\fun\twoup{\Gamma}$ \SUCHTHAT 
 \[\propAequ{q_i\in \mathfrak{d}(\nu)}{\nu\in I_{qi}}.\] %
So far, this discretization does not include any information about its timing, i.e., the formal construction of $\signalmap$. Out of the many different options, we discuss two possible maps $\signalmap_a$ and $\signalmap_b$ as depicted in Figure~\ref{fig:timescale1} and Figure~\ref{fig:timescale2}.\\
First, consider a signal map  $\signalmap_a$ s.t.
for all $ \gamma\in\Gamma^\TE,\tau_a\in\timescale$ and $\w\in\Beh$, it holds that $\Tuple{\gamma,\tau_a}\in\signalmap_a(\w)$ \IFF
\begin{equation*}
 \gamma(0)\in \mathfrak{d}(w(0)),\quad\tau_a^{-1}(0)=\Set{0}
\end{equation*}
and for all $k\in\TE,k>0$, 
\begin{align}
 \tau_a^{-1}(k)&=\left\{\mathrm{glb}\SetCompX{t\geq\tau_a^{-1}(k-1)}{w(t)\notin\mathfrak{d}^{-1}(\gamma(k-1))}\right\}\notag\\
 \gamma(k)&\in \mathfrak{d}(w(\tau_a^{-1}(k))),\label{equ:exp:1:1}
\end{align}
where $\mathrm{glb}$ denotes the greatest lower bound and ${\mathfrak{d}^{-1}(q_i)=I_{qi}}$. This generates the \textbf{point to point time scale transformation} depicted in Figure~\ref{fig:timescale1} (middle), where different points in $\dom{\tau_a}$ are mapped to different points in $\TE$, and an external event is triggered when leaving the interval. The generated external signal $\gamma$ is depicted in Figure~\ref{fig:timescale2}. This map $\signalmap_a$ can be extended to generate a \textbf{set to point time scale transformation} by defining
\begin{equation}\label{equ:exp:1:2}
 \tau_b^{-1}(k)=\left[\tau_a^{-1}(k),\tau_a^{-1}(k+1)\right),
\end{equation}
where every point in $\T$ is in the domain of $\tau_b$. This time scale transformation is depicted in Figure~\ref{fig:timescale1} (bottom). Combining the construction of $\tau_b$ \eqref{equ:exp:1:2} with the construction of $\gamma$ in \eqref{equ:exp:1:1} defines a signal map $\signalmap_b$.\\
Now assume, that we have a signal $\tilde{w}\in\Beh$ that stays in $I_{q_1}$ for all $t$. This signal would only generate one external event $q_1$ at time $0$ but not an infinite sequence of events $\gamma\in\Gamma^\TE$, where $\TE$ is right unbounded. Therefore, the signal maps $\signalmap_a$ and $\signalmap_b$ map $\tilde{w}$ to the empty set.
Obviously, one could repeat the symbol $q_1$ infinitely often to generate a signal in $\gamma\in\Gamma^\TE$ from $\tilde{w}$. However, if one has to know that $w$ will never leave $I_{q_1}$ to do so, as suggested in \cite[Def.7.2]{TabuadaBook}, this generates an anticipating signal map. A non-anticipating version is, for example, obtained, if a symbol is repeated after a fixed time $t_\mathfrak{d}$, if the quantization interval is not left. This would combine event triggered with slow time triggered discretization.
\end{example}
 \begin{figure}[htb]
 \begin{center}
 \begin{tikzpicture}[scale=1]
 \begin{pgfonlayer}{background}
   \path      (0,0) node (o) {
      \includegraphics[width=0.4\linewidth]{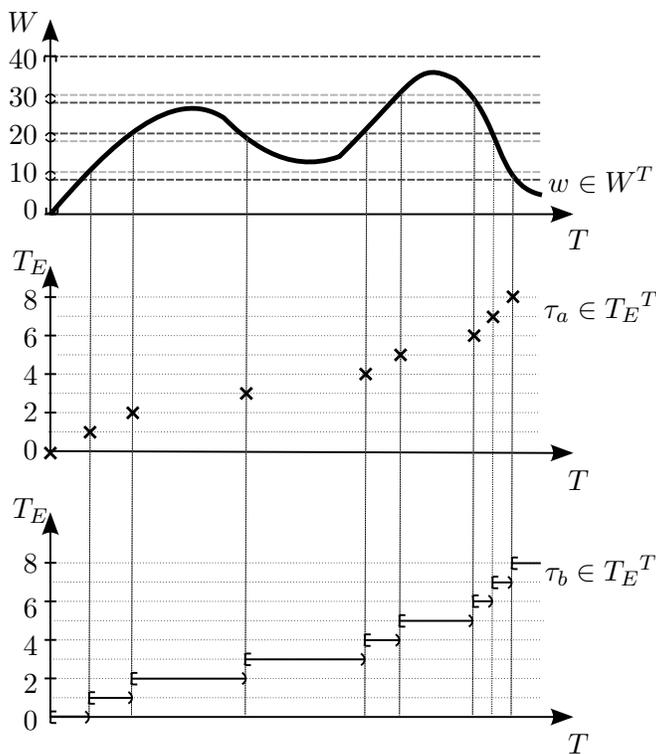}};
 \end{pgfonlayer}
 \begin{pgfonlayer}{foreground}
  \path (o.north west)+(-0.1,-0.2) node {$\W$};
  \path (o.north west)+(-0.1,-0.7) node {$40$};
  \path (o.north west)+(-0.1,-1.2) node {$30$};
  \path (o.north west)+(-0.1,-1.7) node {$20$};
  \path (o.north west)+(-0.1,-2.2) node {$10$};
  \path (o.north west)+(0,-2.7) node (b) {$0$};  
   \path (o.north east)+(0.2,-2.3) node {$w\in\W^\T$};  
   \path (o.north east)+(-0.1,-3.1) node (t2) {$T$};
   \path (b)+(0,-0.7) node {$\TE$};
   \path (b)+(0,-1.2) node {$8$};    
   \path (b)+(0,-1.7) node {$6$};  
   \path (b)+(0,-2.2) node {$4$}; 
   \path (b)+(0,-2.7) node {$2$}; 
   \path (b)+(0,-3.2) node (c) {$0$};
   \path (t2)+(0.3,-0.9) node  {$\tau_a\in\TE^\T$};  
   \path (t2)+(0,-3.2) node  (t3) {$T$};   
   \path (c)+(0,-0.8) node {$\TE$};
   \path (c)+(0,-1.5) node {$8$};    
   \path (c)+(0,-2.1) node {$6$};  
   \path (c)+(0,-2.6) node {$4$}; 
   \path (c)+(0,-3.1) node {$2$}; 
   \path (c)+(0,-3.6) node (c) {$0$};
   \path (t3)+(0.3,-1.2) node  {$\tau_b\in\TE^\T$};  
   \path (t3)+(0,-3.5) node  {$T$};     
 \end{pgfonlayer}
 \end{tikzpicture}
  \vspace{-0.5cm}
 \caption{Illustration of point to point ($\tau_a$) and set to point ($\tau_b$) time scale transformations as constructed in Example~\ref{exp:1}.}\label{fig:timescale1}
 \end{center}
 \end{figure}
\begin{figure}[htb!]
\begin{center}
 \begin{tikzpicture}[scale=1]
 \begin{pgfonlayer}{background}
   \path      (0,0) node (o) {
      \includegraphics[width=0.45\linewidth]{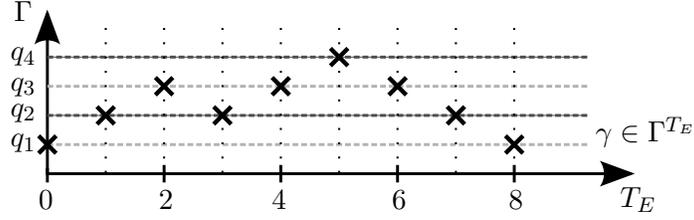}};
 \end{pgfonlayer}
 \begin{pgfonlayer}{foreground}
  \path (o.north west)+(0,-0.2) node (a) {$\Gamma$};
  \path (o.north west)+(0,-0.8) node (b) {$q_4$};
  \path (o.north west)+(0,-1.2) node (b) {$q_3$};
  \path (o.north west)+(0,-1.55) node (b) {$q_2$};
  \path (o.north west)+(0,-1.95) node (b) {$q_1$};  
    \path (o.north east)+(-0.1,-2.7) node (t1) {$\TE$};  
    \path (o.north east)+(0,-1.8) node (t1) {$\gamma\in\Gamma^{\TE}$};
  \path (o.north west)+(0.3,-2.7) node (b) {$0$}; 
  \path (o.north west)+(1.9,-2.7) node (b) {$2$}; 
  \path (o.north west)+(3.4,-2.7) node (b) {$4$};
  \path (o.north west)+(5,-2.7) node (b) {$6$};
  \path (o.north west)+(6.55,-2.7) node (b) {$8$};
 \end{pgfonlayer}
 \end{tikzpicture}
   \vspace{-0.5cm}
 \caption{Illustration of the external signal constructed using event triggered discretization in Example~\ref{exp:1} corresponding to the internal signal depicted in Figure~\ref{fig:timescale1} (top).}\label{fig:timescale2}
 \end{center}
 \end{figure}

 Example~\ref{exp:1} shows that in general $\tau\in\timescale$ is indeed a function of $\w\in\Beh$ when using an event-triggered discretization scheme. Of course, using time-triggered discretization would result in a unique time scale transformation independent from $\w$. Furthermore, the signal maps used in Example~\ref{exp:1} are deterministic in the sense that every signal $w\in\Beh$ generates a one element set 
 or the empty set. However, nondeterministic maps occur for example if $\signalmap$ is constructed from a cover of $\W$ with overlaps of more than two sets.\\

\section{State Space $\signalmap$~-~Dynamical Systems}\label{sec:sysSS}

States are internal variables for which the axiom of state holds, i.e., all relevant information on the past of the system is captured by those variables. In the literature two concepts of the state property exist for behavioral systems. Firstly, the well known version by Willems \cite{Willems1989,Willems1991}, where state trajectories $x_1$ and $x_2$ can be concatenated, if they exhibit the same value at the same time (i.e., 
$\AllQ{x_1,x_2\text{ and }t\in\T}{\propImp{x_1(t)=x_2(t)}{x=\CONCAT{x_1}{}{t}{x_2}}}$ is also a state trajectory). 
And secondly, a generalized version that allows state trajectories to be also concatenated if they reach the same value at different times (i.e., 
$\AllQ{x_1,x_2\text{ and }t_1,t_2\in\T}{\propImp{x_1(t_1)=x_2(t_2)}{x=\CONCAT{x_1}{t_1}{t_2}{x_2}}}$ is also a state trajectory), as used in the context of state maps by Julius and van der Schaft in \cite{JuliusSchaft2005,JuliusPhdThesis2005}. To clearly differentiate both notions we call the first one synchronous and the second one asynchronous. 
Using these two state properties, we construct state space $\signalmap$-dynamical systems such that the discussed state property is preserved by the signal map $\signalmap$.\\

\begin{definition}\label{def:SSDynSysInducesFiniteV}
Let $\Ep{}=\EpRhs$ be a  $\signalmap$-dynamical system, $\X$ be a set and $\BehS{}\subseteq(\WT\times\X)^\T$. \\
Then $\EpS{}=\EpSRhs{}$ is an \textbf{asynchronous state space $\signalmap$-dynamical system} if
\begin{equation}\label{equ:def:SSDynSysInducesFiniteV}
   \AllQSplit{\Tuple{w_1,x_1}\in\BehS{},\Tuple{w_2,x_2}\in\BehS{},t_1,t_2\in\T,\Tuple{\gamma_2,\tau_2}\in\signalmap(w_2),\Tuple{\gamma_1,\tau_1}\in\signalmap(w_1),k_1,k_2\in\TE}{\propImp{
   \begin{propConjA}
    x_1(t_1)=x_2(t_2)\\
    k_1=\timescaleUp{\tau_1}(t_1)\\
    k_2=\timescaleUp{\tau_2}(t_2)
   \end{propConjA}}{
\begin{propConjA}
 \CONCAT{\Tuple{w_1,x_1}}{t_1}{t_2}{\Tuple{w_2,x_2}}\in\BehS{}\\
\Tuple{\CONCAT{\gamma_1}{k_1}{k_2}{\gamma_2},\CONCAT{\tau_1}{t_1}{t_2}{\BR{\tau_2+c}}}\in\signalmap(\CONCAT{w_1}{t_1}{t_2}{w_2})
\end{propConjA},}}
\end{equation}
where $\AllQ{t\in\T}{c(t)=k_1-k_2}$.
Furthermore, $\EpS{}$ is an \textbf{externally synchronous} state space $\signalmap$-dynamical system if \eqref{equ:def:SSDynSysInducesFiniteV} holds for $k=k_1=k_2$
and a \textbf{synchronous} state space $\signalmap$-dynamical system if \eqref{equ:def:SSDynSysInducesFiniteV} holds for $t=t_1=t_2$ and $k=k_1=k_2$.
\end{definition}

It is easy to see that every asynchronous state space $\signalmap$-dynamical system is also an externally synchronous and a synchronous one, because we can always pick $k=k_1=k_2$ and $t=t_1=t_2$ in \eqref{equ:def:SSDynSysInducesFiniteV}. With the same argument, every externally synchronous state space $\signalmap$-dynamical system is also a synchronous one. 
For the asynchronous and the synchronous case in \REFdef{def:SSDynSysInducesFiniteV}, the implication  $\propImp{x_1(t_1)=x_2(t_2)}{ \CONCAT{\Tuple{w_1,x_1}}{t_1}{t_2}{\Tuple{w_2,x_2}}\in\BehS{}}$ is equivalent to the asynchronous and for $t=t_1=t_2$ to the synchronous state property for the system $\ES{}=\Tuple{\T,\W\times\X,\BehS{}}$. 
The additional requirement in \eqref{equ:def:SSDynSysInducesFiniteV} ensures, that this concatenation property also holds for the external behavior.
Note that for the externally synchronous case, synchronization is only required on the external time axis.\\ 
In the remainder of this paper, we refer to a system as introduced in  \REFdef{def:SSDynSysInducesFiniteV} simply as \textit{state space $\signalmap$-dynamical system}, if the respective adjective (asynchronous, externally synchronous, synchronous) is irrelevant.\\
Since possibly not all states are reachable by a state trajectory in $\projState{\X}{\BehS{}}$, we define the following reachable subsets of the state space (comp.\cite[Def.5.37]{JuliusPhdThesis2005}).\\

\begin{definition}\label{def:TimeIndexedStateSpace1}
Let $\EpS{}=\EpSRhs{}$ be a state space $\signalmap$-dynamical system. Then 
\begin{align*}
 &\XT{}\deff\bigcup_{t\in\T}\Xt{t}{}\quad\text{and}\quad \XK{}\deff\bigcup_{k\in\T_E}\Xk{k}{}\quad	\text{s.t.}\\
 &\Xt{t}{}\deff\SetCompX{\xi\in\X}{\ExQ{\Tuple{w,x}\in\BehS{}}{x(t)=\xi}} \quad \text{and}\\
 &\Xk{k}{}\deff\SetCompX{\xi\in\X}{
 \begin{small}
  \ExQSplit{\Tuple{w,x}\in\BehS{},\Tuple{\gamma,\tau}\in\signalmap(w),t\in\timescaleDown{\tau}(k)}{x(t)=\xi}
 \end{small} }
\end{align*}
are the internal and external time-indexed state spaces $\XT{}\subseteq\X$ and $\XK{}\subseteq\X$, respectively.
\end{definition}
Obviously, the internal and external time-indexed state spaces are equivalent if $\timescaleUp{\tau}$ is a total function.\\

\section{Simulation Relations}\label{sec:SimRel}

One system simulates another one, if its external behavior contains the external behavior
of the latter, while ensuring that the state trajectories generated by both systems only visit states, at each instant of time, that are associated by a relation. 
To formalize this property, a special relation, called \textit{simulation relation}, is constructed between both state spaces.\\
In the behavioral framework signals are usually  right-unbounded. It is well known that a local (i.e., on a finite time interval) evaluation of properties is only possible, if the system is complete \cite{Willems1989}. 
Inspired by \cite[Def. 5.21]{JuliusPhdThesis2005}, we therefore define a concatenation based simulation relation for $\signalmap$-dynamical systems. In contrast to the locally defined simulation relation used for transition systems (e.g., in \cite{Pappas2003,TabuadaPappas2003b}) or general flow systems (in \cite{DavorenTabuada2007}), it also relates not necessarily complete systems.\\

\begin{definition}\label{def:SimRel}
 Let $\EpS{1}=\EpSRhs{1}$ and $\EpS{2}=\EpSRhs{2}$ be state space $\signalmap$-dynamical systems.\\ 
Then a relation $\R\subseteq\X_1\times\X_2$ is an \textbf{asynchronous simulation relation} from $\EpS{1}$ to $\EpS{2}$  (written $\R\in\SR{\async}{\EpS{1}}{\EpS{2}}$) , i.e., $\EpS{2}$ asynchronously simulates $\EpS{1}$, if
\begin{subequations}\label{equ:def:SimRel}
 \begin{equation}\label{equ:def:SimRel:a}
  \AllQ{\xi_1\in\XK{1}}{\ExQ*{\xi_2\in\XK{2}}{\Tuple{\xi_1,\xi_2}\in\R}}
 \end{equation}
and
  \begin{small}
 \begin{equation}\label{equ:def:SimRel:b}
 \AllQSplit{\Tuple{w_1,x_1}\in\BehS{1},\Tuple{w',x'}\in\BehS{2}, \Tuple{\gamma_1,\tau_1}\in\signalmap_1(w_1),\Tuple{\gamma',\tau'}\in\signalmap_2(w'), t_1\in\T_1,t_2\in\T_2, k_1,k_2\in\T_E}{
\propImp{
\begin{propConjA}
\Tuple{x_1(t_1),x'(t_2)}\in\R\\
k_1=\timescaleUp{\tau_1}(t_1)\\
k_2=\timescaleUp{\tau'}(t_2)
\end{propConjA}
}{
\ExQSplit{\Tuple{w_2,x_2}\in\BehS{2}, \Tuple{\gamma_2,\tau_2}\in\signalmap_2(w_2)}{
\begin{propConjA}
\gamma_2=\CONCAT{\gamma'}{k_2}{k_1}{\gamma_1}\\
\AllQ{t\in\T_2, t< t_2}{
\begin{propConjA}
w_2(t)=w'(t)\\
x_2(t)=x'(t)\\
\tau_2(t)=\tau'(t)\\
\end{propConjA}
}\\
x_2(t_2)=x'(t_2)\\
\AllQSplit{k\geq k_2,t_1'\in\timescaleDown{\tau_1}(k-k_2+k_1),t_1'>t_1}{\ExQ{t_2'\in\timescaleDown{\tau_2}(k),t_2'>t_2 }{
\Tuple{x_1(t_1'),x_2(t_2')}\in\R}}
\end{propConjA}.
}
}} 
\end{equation}
  \end{small}
\end{subequations}
It is an \textbf{externally synchronous simulation relation} from $\EpS{1}$ to $\EpS{2}$  (written $\R\in\SR{\wsync}{\EpS{1}}{\EpS{2}}$) if
\begin{subequations}\label{equ:def:SimRelSync}
\begin{equation}\label{equ:def:SimRelSync:a}
  \AllQ{k\in\TE,\xi_1\in\Xk{k}{1}}{\ExQ*{\xi_2\in\Xk{k}{2}}{\Tuple{\xi_1,\xi_2}\in\R}}
 \end{equation}
 \end{subequations}
and \eqref{equ:def:SimRel:b} holds for $k=k_1=k_2$.\\
Furthermore, if $\T=\T_1=\T_2$, then $\R$ is a \textbf{ synchronous simulation relation} from $\EpS{1}$ to $\EpS{2}$ (written $\R\in\SR{\sync}{\EpS{1}}{\EpS{2}}$) if 
\begin{subequations}\label{equ:def:SimRelSync:2}
\begin{equation}\label{equ:def:SimRelSync:2:a}
  \AllQ{t\in\T,\xi_1\in\Xt{t}{1}}{\ExQ*{\xi_2\in\Xt{t}{2}}{\Tuple{\xi_1,\xi_2}\in\R}}
 \end{equation}
 \end{subequations}
and \eqref{equ:def:SimRel:b} holds for $k=k_1=k_2$ and $t=t_1=t_2$.
\end{definition}

\begin{remark}
The construction of the externally synchronous simulation relation in \REFdef{def:SimRel} is inspired by the so called \textit{synchronized simulation relation} defined in \cite[Def. 5.38]{JuliusPhdThesis2005}. However, the latter does not restrict \eqref{equ:def:SimRel:b} to hold only for $k=k_1=k_2$. 
\end{remark}

The intuitive interpretation of the terms asynchronous, synchronous and externally synchronous is strongly related to the ones used in \REFdef{def:SSDynSysInducesFiniteV}. However, in \REFdef{def:SimRel} the synchronization takes place between signals of \textit{different} systems that are related. 
\\
In contrast to \REFdef{def:SSDynSysInducesFiniteV}, it is not true that every asynchronous simulation relation is an (externally) synchronous one, since \eqref{equ:def:SimRel:a} does generally not imply \eqref{equ:def:SimRelSync:a} and \eqref{equ:def:SimRelSync:2:a}.  
Intuitively, if $\R$ is an asynchronous simulation relation, we know that \eqref{equ:def:SimRel:b} holds for $t=t_1=t_2$ and $k=k_1=k_2$. However, we can generally not ensure that for every state in $\XK{1}$ reachable at external time $k$ and internal time $t$, there exists a related state in $\XK{2}$ that is reachable at the same external and internal time. We can therefore possibly not relate the whole state space in a synchronous or externally synchronous fashion, implying that  $\R$ may formally not be an (externally) synchronous simulation relation.\\
To generate some intuition for the simulation relation constructed in \REFdef{def:SimRel}, we will discuss \eqref{equ:def:SimRel:b} using some graphical illustrations. For this purpose assume that we have signals $\Tuple{w_1,x_1}\in\BehS{1},\Tuple{w',x'}\in\BehS{2}, \Tuple{\gamma_1,\tau_1}\in\signalmap_1(w_1)$ and $\Tuple{\gamma',\tau'}\in\signalmap_2(w')$ such that the states $\xi_1=x_1(t_1)$ and $\xi_2=x'(t_2)$, with $k_1=\timescaleUp{\tau_1}(t_1)$ and $k_2=\timescaleUp{\tau_2}(t_2)$, are related. To simulate $\EpS{1}$, the system $\EpS{2}$ must be able to continue from time $k_2$ with the same external signal as produced by $\EpS{1}$ after $k_1$. This is expressed in \eqref{equ:def:SimRel:b} by requiring the existence of an external signal $\gamma_2\in\BehE{2}$ which is constructed from the concatenation of the signals $\gamma'$ and $\gamma_1$, as depicted in Figure~\ref{fig:simrel_gamma}. 

\begin{figure}[htb]
\begin{center}
 \begin{tikzpicture}[scale=1]
 \begin{pgfonlayer}{background}
   \path      (0,0) node (o) {
      \includegraphics[width=0.45\linewidth]{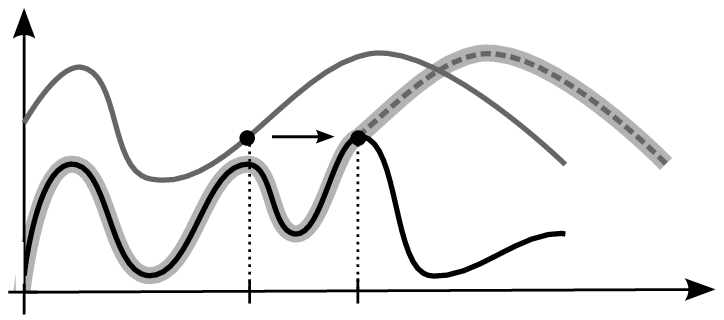}};
 \end{pgfonlayer}
 \begin{pgfonlayer}{foreground}
  \begin{footnotesize}
  \path (o.north west)+(0,-0.2) node (a) {$\Gamma$};
  \path (o.north east)+(-0.3,-1.9) node (t1) {$\gamma_2$};
\path (o.north east)+(-1.4,-1.9) node (t2) {$\gamma_1$}; 
\path (o.north east)+(-1.4,-2.5) node (t2) {$\gamma'$}; 
   \path (o.north east)+(0,-3.5) node (t2) {$\TE$};
 \path (o.south west)+(2.9,0) node (k2) {$k_1$}; 
 \path (o.south west)+(4.1,0) node (k2) {$k_2$}; 
\end{footnotesize}
 \end{pgfonlayer}
 \end{tikzpicture}
 \end{center}
 \caption{Visualization of the concatenation ${\gamma_2=\CONCAT{\gamma'}{k_2}{k_1}{\gamma_1}}$ in \eqref{equ:def:SimRel:b}. }\label{fig:simrel_gamma}
 \end{figure}
 
 To ensure that $\signalmap_2$ is non-anticipating, this concatenation is not allowed to change the past, which is why we require that the past of $x_2,w_2$ and $\tau_2$ match the past\footnote{
 In contrast to \cite[Def. 5.21]{JuliusPhdThesis2005}, we only require the strict past to match, because our concatenation definition \eqref{equ:concat} slightly differs from the one used in \cite{JuliusPhdThesis2005}.}
 of $x',w'$ and $\tau'$. Moreover, we have to ensure, that the state trajectories $x'$ and $x_2$ match at time $t_2$, expressed by $x_2(t_2)=x'(t_2)$.\\
\begin{figure}[htb]
\begin{center}
 \begin{tikzpicture}[scale=1]
 \begin{pgfonlayer}{background}
   \path      (0,0) node (o) {
      \includegraphics[width=0.45\linewidth]{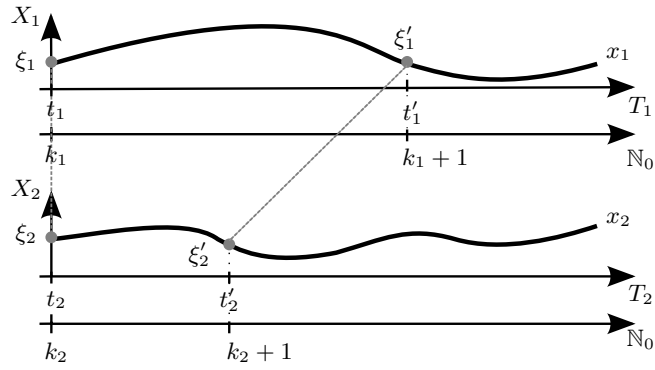}};
 \end{pgfonlayer}
 \begin{pgfonlayer}{foreground}
  \begin{footnotesize}
  \path (o.north west)+(0,-0.2) node (a) {$\X_1$};
  \path (o.north west)+(0,-0.8) node (b) {$\xi_1$};
   \path (o.north west)+(0,-2.5) node (c) {$\X_2$};
  \path (o.north west)+(0,-3.1) node (d) {$\xi_2$};
     \path (o.north east)+(-0.4,-0.7) node {$x_1$}; 
    \path (o.north east)+(-0.1,-1.4) node (t1) {$\T_1$};
    \path (o.north east)+(-0.1,-2.1) node (k1) {$\Nbn$};
\path (o.north east)+(-0.4,-2.9) node (t2) {$x_2$};
   \path (o.north east)+(-0.1,-3.9) node (t2) {$\T_2$};
   \path (o.north east)+(-0.1,-4.5) node (k2) {$\Nbn$};

    \path (0.9,1.8) node (xa) {$\xi_1'$};
    \path (o.north west)+(0.4,-1.45) node (d) {$t_1$};
     \path (o.north west)+(0.4,-2.05) node (d) {$k_1$};  
    \path (1,0.8) node (xa) {$t_1'$};
    \path (1.3,0.2) node (xa) {$k_1+1$};
    \path (o.south west)+(0.4,0.6) node (k2) {$t_2$}; 
     \path (o.south west)+(0.4,-0.1) node (k2) {$k_2$};
    \path (o.south west)+(2.7,0.6) node (k2) {$t_2'$};    
    \path (o.south west)+(3.1,-0.1) node (k2) {$k_2+1$}; 
     \path (o.south west)+(2.3,1.2) node (k2) {$\xi_2'$};    
   \end{footnotesize}   
 \end{pgfonlayer}
 \end{tikzpicture}
 \end{center}
  \vspace{-0.3cm}
  \caption{Visualization of the last line in \eqref{equ:def:SimRel:b} for two point to point time scale transformations $\tau_1$ and $\tau_2$, with $\timescaleDown{\tau_1}(k_1)=\Set{t_1}$, $\timescaleDown{\tau_1}(k_1+1)=\Set{t_1'}$,  $\timescaleDown{\tau_2}(k_2)=\Set{t_2}$ and $\timescaleDown{\tau_2}(k_2+1)=\Set{t_2'}$. Gray lines connect related states.}\label{fig:r-simulation_hscc}
\end{figure}

\begin{figure}[htb]
\begin{center}
 \begin{tikzpicture}[scale=1]
 \begin{pgfonlayer}{background}
   \path      (0,0) node (o) {
      \includegraphics[width=0.45\linewidth]{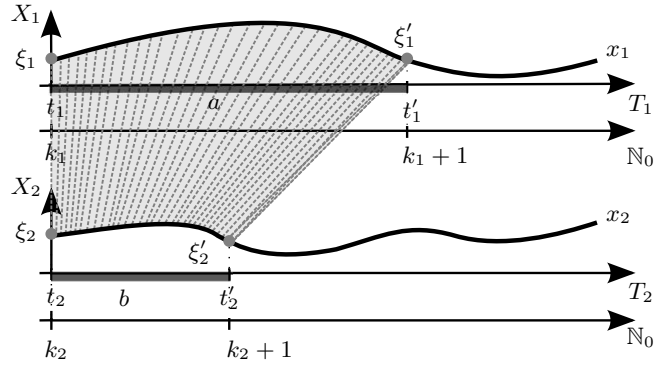}};
 \end{pgfonlayer}
 \begin{pgfonlayer}{foreground}
 \begin{footnotesize}
  \path (o.north west)+(0,-0.2) node (a) {$\X_1$};
  \path (o.north west)+(0,-0.8) node (b) {$\xi_1$};
   \path (o.north west)+(0,-2.5) node (c) {$\X_2$};
  \path (o.north west)+(0,-3.1) node (d) {$\xi_2$};
     \path (o.north east)+(-0.4,-0.7) node {$x_1$}; 
    \path (o.north east)+(-0.1,-1.4) node (t1) {$\T_1$};
    \path (o.north east)+(-0.1,-2.1) node (k1) {$\Nbn$};
\path (o.north east)+(-0.4,-2.9) node (t2) {$x_2$};
   \path (o.north east)+(-0.1,-3.9) node (t2) {$\T_2$};
   \path (o.north east)+(-0.1,-4.5) node (k2) {$\Nbn$};

    \path (0.9,1.8) node (xa) {$\xi_1'$};
    \path (o.north west)+(2.5,-1.4) node (d) {$a$};
    \path (o.north west)+(0.4,-1.45) node (d) {$t_1$};
     \path (o.north west)+(0.4,-2.05) node (d) {$k_1$};  
    \path (1,0.8) node (xa) {$t_1'$};
    \path (1.3,0.2) node (xa) {$k_1+1$};
    \path (o.south west)+(1.3,0.6) node (k2) {$b$}; 
    \path (o.south west)+(0.4,0.6) node (k2) {$t_2$}; 
     \path (o.south west)+(0.4,-0.1) node (k2) {$k_2$};
    \path (o.south west)+(2.7,0.6) node (k2) {$t_2'$};    
    \path (o.south west)+(3.1,-0.1) node (k2) {$k_2+1$}; 
     \path (o.south west)+(2.3,1.2) node (k2) {$\xi_2'$};    
  \end{footnotesize}   
 \end{pgfonlayer}
 \end{tikzpicture}
 \end{center}
  \vspace{-0.3cm}
  \caption{Visualization of the last line in \eqref{equ:def:SimRel:b} for two set to point time scale transformations $\tau_1$ and $\tau_2$, with $a=\timescaleDown{\tau_1}(k_1)=[t_1,t_1')$ and $b=\timescaleDown{\tau_2}(k_2)=[t_2,t_2')$. Gray lines connect related states.}\label{fig:p-simulation_hscc}
\end{figure}

\begin{figure}[htb]
\begin{center}
 \begin{tikzpicture}
 \begin{pgfonlayer}{background}
   \path      (0,0) node (o) {
      \includegraphics[width=0.45\linewidth]{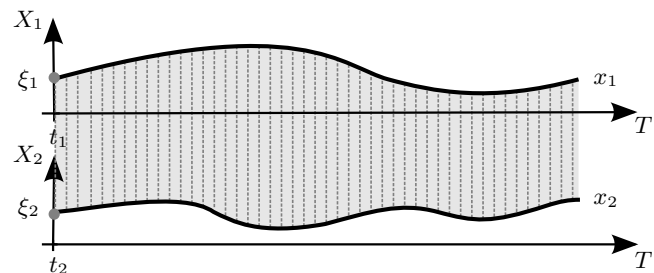}};
 \end{pgfonlayer}
 \begin{pgfonlayer}{foreground}
  \begin{footnotesize}
  \path (o.north west)+(0,-0.2) node (a) {$\X_1$};
  \path (o.north west)+(0,-1) node (b) {$\xi_1$};
   \path (o.north west)+(0,-2) node (c) {$\X_2$};
  \path (o.north west)+(0,-2.7) node (d) {$\xi_2$};
  \path (o.north east)+(-0.6,-1) node (t1) {$x_1$};
    \path (o.north east)+(-0.1,-1.6) node (t1) {$\T$};
\path (o.north east)+(-0.6,-2.6) node (t2) {$x_2$};    
   \path (o.north east)+(-0.1,-3.4) node (t2) {$\T$};
  \path (o.south west)+(0.4,1.7) node (k2) {$t_1$};  
 \path (o.south west)+(0.4,0) node (k2) {$t_2$};  
\end{footnotesize}
 \end{pgfonlayer}
 \end{tikzpicture}
 \end{center}
 \vspace{-0.7cm}
 \caption{Visualization of the last line in \eqref{equ:def:SimRel:b} for two identity time scale transformations $\tau_1=\tau_2=\I$.}\label{fig:t-simulation_hscc}
 \end{figure}
 The last line of \eqref{equ:def:SimRel:b} basically says that the state trajectories $x_1$ and $x_2$ need to stay related for all future external time instants. 
However, the nature of $\tau$ significantly influences how restrictive this requirement is. For example, having a point to point time scale transformation in both systems only requires state trajectories to be related at sampling points (Figure~\ref{fig:r-simulation_hscc}), while a set to point time scale transformation, for example, requires state trajectories to be related at all future times (Figure~\ref{fig:p-simulation_hscc}). However, as clearly visible in Figure~\ref{fig:r-simulation_hscc} and \ref{fig:p-simulation_hscc}, both cases allow for a stretching or shrinking of time between related state trajectories. If both systems have an identity time scale transformation (and therefore $\T=\T_1=\T_2=\TE$) this stretching or shrinking of time is 
no longer allowed, as shown in Figure~\ref{fig:t-simulation_hscc}. Note that the latter case only implies that the constructed asynchronous simulation relation is also a synchronous one, if we additionally require $k=k_1=k_2$, which immediately implies $t=t_1=t_2$.\\

 \begin{remark}\label{rem:GFS}
The intuitive interpretation of the different simulation relations depicted in Figure~\ref{fig:r-simulation_hscc}~-~\ref{fig:t-simulation_hscc} is very similar to the idea behind the $r$-, $p$- and $t$-simulation relations constructed in \cite{DavorenTabuada2007} for general flow systems. This suggests that for the subclass of discrete and continuous systems, our simulation relation can reproduce the relations in  \cite{DavorenTabuada2007} by choosing different time scale transformations. However, our relation extends the constructions in \cite{DavorenTabuada2007} by allowing to include the simulation of external trajectories. Furthermore, relating two systems with different time scale transformations gives an even richer variety of relations.
\end{remark}

\begin{remark}\label{rem:TS}
Recall that $\signalmap$-dynamical systems capture the dynamics of transition systems and linear time-invariant continuous systems (see Remark~\ref{rem:0}), if $\T=\T_1=\T_2=\TE$. 
Relating two systems implies a state trajectory matching requirement as depicted in Figure~\ref{fig:t-simulation_hscc}. Additionally, the external signal $\gamma$, which, in the case of transition systems is the output $y$, and, in the case of linear time-invariant continuous systems is the pair $\Tuple{u,y}$, needs to satisfy the requirement depicted in Figure~\ref{fig:simrel_gamma}. Observe that for complete systems this interpretation coincides with the locally defined simulation relation for transition systems, e.g., in \cite{Pappas2003,TabuadaPappas2003,TabuadaBook}. The same is true for the simulation relation constructed for linear time invariant systems in \cite{Schaft2004}. This suggests that both notions of simulation relations can be captured by our notion. 
\end{remark}

Using the simulation relations constructed in \REFdef{def:SimRel}, we can define similarity and bisimilarity for the class of state space $\signalmap$-dynamical systems in the usual fashion.\\
\begin{definition}\label{def:BisimRel_Behaviour}
$\EpS{1}$ is \textbf{asynchronously simulated}  by $\EpS{2}$, denoted by ${\EpS{1}\kg{\async}\EpS{2}}$, if there exists an asynchro\-nous simulation relation from $\ES{1}$ to $\ES{2}$. 
$\EpS{1}$ and $\EpS{2}$ are \textbf{asynchronously bisimilar}, denoted by ${\EpS{1}\hg{\async}\EpS{2}}$, if there exists a relation $\R\subseteq\X_1\times\X_2$ \SUCHTHAT $R$ and  $R^{-1}=\SetComp{\Tuple{x_2,x_1}}{\Tuple{x_1,x_2}\in\R}$ are asynchronous simulation relations from $\ES{1}$ to $\ES{2}$ and from $\ES{2}$ to $\ES{1}$, respectively.\\
$\EpS{1}$ is \textbf{externally synchronously simulated}  by $\EpS{2}$, denoted by ${\EpS{1}\kg{\wsync}\EpS{2}}$, if there exists an externally  synchronous simulation relation from $\ES{1}$ to $\ES{2}$. 
$\EpS{1}$ and $\EpS{2}$ are \textbf{externally synchronously bisimilar}, denoted by ${\EpS{1}\hg{\wsync}\EpS{2}}$, if there exists a relation $\R\subseteq\X_1\times\X_2$ \SUCHTHAT $R$ and $R^{-1}$ are externally synchronous simulation relations from $\ES{1}$ to $\ES{2}$ and from $\ES{2}$ to $\ES{1}$, respectively.\\
$\EpS{1}$ is \textbf{synchronously simulated}  by $\EpS{2}$, denoted by  ${\EpS{1}\kg{\sync}\EpS{2}}$, if there exists a synchronous simulation relation from $\ES{1}$ to $\ES{2}$. 
$\EpS{1}$ and $\EpS{2}$ are \textbf{synchronously bisimilar}, denoted by ${\EpS{1}\hg{\sync}\EpS{2}}$, if there exists a relation $\R\subseteq\X_1\times\X_2$ \SUCHTHAT $R$ and $R^{-1}$ are synchronous simulation relations from $\ES{1}$ to $\ES{2}$ and from $\ES{2}$ to $\ES{1}$, respectively.
\end{definition}

\section{Equivalence of External Behaviors}
Before proving the soundness of our construction we introduce another simulation relation to discuss the connection between behavioral equivalence and bisimilarity of two systems.\\

\begin{definition}\label{def:InitSimRel}
Let $\EpS{1}=\EpSRhs{1}$ and $\EpS{2}=\EpSRhs{2}$ be state space $\signalmap$-dynamical systems and let $l\in\TE$.\\
Then $\R\subseteq\X_1\times\X_2$ is an \textbf{$l$-initial simulation relation} from $\EpS{1}$ to $\EpS{2}$ (written $\R\in\SR{l}{\EpS{1}}{\EpS{2}}$) if
\begin{subequations}\label{equ:def:InitSimRel}
 \begin{equation}\label{equ:def:InitSimRel:a}
  \AllQ{\xi_1\in\Xk{l}{1}}{\ExQ*{\xi_2\in\Xk{l}{2}}{\Tuple{\xi_1,\xi_2}\in\R}}
 \end{equation}
  \end{subequations}
and \eqref{equ:def:SimRel:b} holds.
\end{definition}

For this simulation relation, $l$-initially similar and bisimilar systems are defined analogously to \REFdef{def:BisimRel_Behaviour} and are denoted by $\EpS{1}\kg{l}\EpS{2}$ and $\EpS{1}\hg{l}\EpS{2}$, respectively.\\
Observe that in \REFdef{def:InitSimRel}, the statement in \eqref{equ:def:SimRel:b} still needs to hold for arbitrary $k_1,k_2$ and $t_1,t_2$ (as for the asynchronous simulation relation). However, we  require in \eqref{equ:def:InitSimRel:a} that states $\xi_1$ reached at external time $k=l$ are related to states $\xi_2$ also reachable at external time $k=l$. Observe that this does in general not imply that \eqref{equ:def:SimRel:a} holds.
Due to the iterative nature of \eqref{equ:def:SimRel:b}, intuitively, relating states reached at external time $k=l$ leads to a relation between all states reachable for external time $k>l$ (explaining the name for this relation). In particular, if the external time axis has a minimal element $\nu$ (e.g., $\TE=\Nbn$ or $\TE=\Rbn$ with $\nu=0$), a $\nu$-initial simulation relation will imply that all reachable states are related in an externally synchronized fashion. The following lemma formalizes this intuition by proving various connections between the different relation types.\\

\begin{lemma}\label{lem:LinitSimImpliesOtherSims}
Let $\EpS{1}=\EpSRhs{1}$ and $\EpS{2}=\EpSRhs{2}$ be state-space $\signalmap$-dynamical systems \SUCHTHAT $\TE$ has the minimal element $\nu$. Then
 \begin{compactenum}[(i)]
   \item $\propImp{\R\in\SR{l=\nu}{\EpS{1}}{\EpS{2}}}{
   \R\in\SR{\wsync}{\EpS{1}}{\EpS{2}}}$,
  \item $\propImp{\R\in\SR{l=\nu}{\EpS{1}}{\EpS{2}}}{
   \R\in\SR{\async}{\EpS{1}}{\EpS{2}}}$, and
  \item\hspace{-0.2cm} ${\propImp{
  \begin{propConjA}
  \R\in\SR{l=\nu}{\EpS{1}}{\EpS{2}}\\
  \T_1=\T_2=\TE\\
  \AllQ{w_1,\Tuple{\gamma_1,\tau_1}\hspace{-0.1cm}\in\hspace{-0.1cm}\signalmap_1(w_1)}{\tau_1\hspace{-0.1cm}=\hspace{-0.1cm}\I}\\
  \AllQ{w_2,\Tuple{\gamma_2,\tau_2}\hspace{-0.1cm}\in\hspace{-0.1cm}\signalmap_2(w_2)}{\tau_2\hspace{-0.1cm}=\hspace{-0.1cm}\I}
  \end{propConjA}
  }{{\R\hspace{-0.08cm}\in\hspace{-0.08cm}\SR{\sync}{\EpS{1}}{\EpS{2}}}.}}$
 \end{compactenum}
\end{lemma}

\begin{proof}
Pick $\R\in\SR{l=\nu}{\EpS{1}}{\EpS{2}}$ and observe the following facts:
\begin{compactenum}[ (A)]
\item \label{item:proof:lem:LinitSimImpliesOtherSims:1} %
 \eqref{equ:def:SimRelSync:a} holds for $\R$:\\
As \eqref{equ:def:InitSimRel:a} holds for $\R$ (using \REFdef{def:TimeIndexedStateSpace1}) we can fix $\Tuple{w_1,x_1}\in\BehS{1}, \Tuple{\gamma_1,\tau_1}\in\signalmap_1(w_1),t_1\in\timescaleDown{\tau_1}(\nu)$ and $\Tuple{w',x'}\in\BehS{2},\Tuple{\gamma',\tau'}\in\signalmap_2(w'),t_2\in\timescaleDown{\tau_2}(\nu)$ s.t. $\Tuple{x_1(t_1),x'(t_2)}\in\R$. Since $\R\in\SR{l}{\EpS{1}}{\EpS{2}}$, \eqref{equ:def:SimRel:b} implies that there exist $\Tuple{w_2,x_2}\in\BehS{2},\Tuple{\gamma_2,\tau_2}\in\signalmap_2(w_2)$ s.t. $\AllQ{k\geq \nu,t_1'\in\timescaleDown{\tau_1}(k)}{\ExQ{t_2'\in\timescaleDown{\tau_2}(k)}{
\Tuple{x_1(t_1'),x_2(t_2')}\in\R}}$. Using \REFdef{def:TimeIndexedStateSpace1} and the fact that $\nu$ is the minimal element of $\TE$, this implies that \eqref{equ:def:SimRelSync:a} holds.
 \item \label{item:proof:lem:LinitSimImpliesOtherSims:2} \eqref{equ:def:SimRelSync:a} implies \eqref{equ:def:SimRel:a} since $\XK{i}\deff\bigcup_{k\in\TE}\Xk{k}{i}$ for $i\in\Set{1,2}$ from \REFdef{def:TimeIndexedStateSpace1}.
 \item \label{item:proof:lem:LinitSimImpliesOtherSims:3} Let $i\in\Set{1,2}$. If $\T_i=\TE$ and $\AllQ{w_i,\Tuple{\gamma_i,\tau_i}\in\signalmap_i(w_i)}{\tau_i=\I}$ then $\AllQ{t\in\T_i}{\Xt{t}{i}=\Xk{t}{i}}$ from \REFdef{def:TimeIndexedStateSpace1} implies \eqref{equ:def:SimRelSync:a} \IFF \eqref{equ:def:SimRelSync:2:a}.
 \item \label{item:proof:lem:LinitSimImpliesOtherSims:4} If \eqref{equ:def:SimRel:b} holds, it also holds for $k=k_1=k_2$ and $t=t_1=t_2$.
\end{compactenum} 
Now 
\begin{inparaenum}[ (i)]
 \item follows from \eqref{item:proof:lem:LinitSimImpliesOtherSims:1} and \eqref{item:proof:lem:LinitSimImpliesOtherSims:4}, 
 \item follows from \eqref{item:proof:lem:LinitSimImpliesOtherSims:1} and \eqref{item:proof:lem:LinitSimImpliesOtherSims:2}, and
 \item follows from \eqref{item:proof:lem:LinitSimImpliesOtherSims:1}, \eqref{item:proof:lem:LinitSimImpliesOtherSims:3} and \eqref{item:proof:lem:LinitSimImpliesOtherSims:4}.
\end{inparaenum}
\end{proof}\vspace{0.5cm}
\begin{remark}
 The inverse implication in \REFlem{lem:LinitSimImpliesOtherSims} (i) does not hold, as $\R\in\SR{\wsync}{\EpS{1}}{\EpS{2}}$ does not imply that \eqref{equ:def:SimRel:b} holds for arbitrary $k_1\neq k_2$.
\end{remark}

\begin{remark}
 Recall from Remark~\ref{rem:0} that $\signalmap$-dynamical systems can represent transition systems using an external time axis $\TE=\Nbn$ (with minimal element $\nu=0$). For this system class, simulation relations are usually defined by requiring that the initial states are related and a local property, similar to \eqref{equ:def:SimRel:b}, holds (see, e.g., \cite{Pappas2003,TabuadaPappas2003b,TabuadaBook}). This suggests, that simulation relations defined for transition systems are $0$-initial simulation relations in our sense.
\end{remark}

As the main result of this section we generalize the results in 
\cite[Thm. 5.41]{JuliusPhdThesis2005} to state space $\signalmap$-dynamical systems with external time axis having the minimal element $\nu$ and
show that the existence of a $\nu$-initial simulation relation from one system to another one implies that the behavior of the first is a subset of the second one. As an immediate consequence, behavioral equivalence is obtained if two systems are $\nu$-initially bisimilar.\\

\begin{theorem}\label{thm:BisimImplExtBehInclusion}
Let $\EpS{1}=\EpSRhs{1}$ and $\EpS{2}=\EpSRhs{2}$ be state-space $\signalmap$-dynamical systems \SUCHTHAT $\TE$ has the minimal element $\nu$. 
Then
\begin{compactenum}[(i)]
 \item ${\propImp{\BR{\ES{1}\kg{l=\nu}\ES{2}}}{\BR{\BehE{1}\subseteq\BehE{2}}}}$
\item ${\propImp{\BR{\ES{1}\hg{l=\nu}\ES{2}}}{\BR{\BehE{1}=\BehE{2}}}}$
\end{compactenum}
\end{theorem}

\begin{proof}
Using \eqref{equ:BehE}, the statement $\BehE{1}\subseteq\BehE{2}$ is equivalent to
\[\AllQ{\gamma\in\Gamma^\TE}{
\propImpSplit{\ExQ{\Tuple{\x_1,\w_1}\in\BehS{1},\tau_1\in \timescale_1}{\Tuple{\gamma,\tau_1}\in\signalmap_1(\w_1)}}{\ExQ{\Tuple{\x_2,\w_2}\in\BehS{2},\tau_2\in \timescale_2}{\Tuple{\gamma,\tau_2}\in\signalmap_2(\w_2)},  }}
\]
where $\timescale_i,~i\in\{1,2\}$ is the set of valid time scale transformations from $T_i$ to $\TE$.
Fix $\gamma,\x_1,\w_1,\tau_1$ \SUCHTHAT $\Tuple{\gamma,\tau_1}\in\signalmap_1(\w_1)$.
Since $\ES{1}\kg{l=\nu}\ES{2}$, \eqref{equ:def:InitSimRel:a} holds for $k=\nu$.
Therefore, we can pick $t_1\in\timescaleDown{\tau_1}(\nu),\Tuple{w',x'}\in\BehS{2},\Tuple{\gamma',\tau'}\in\signalmap_2(w'),t_2\in\timescaleDown{\tau'}(\nu)$ s.t.
$\Tuple{x_1(t_1),\x'(t_2)}\in\R$.
Using \eqref{equ:def:SimRel:b} for $k_1=k_2=\nu$ this implies that
$\ExQ{\Tuple{w_2,x_2}\in\BehS{2}, \Tuple{\gamma_2,\tau_2}\in\signalmap_2(w_2)}{
\gamma_2=\CONCAT{\gamma'}{\nu}{\nu}{\gamma}=\gamma}$, which proves statement (i).
 Part (ii) follows immediately from (i) and \REFdef{def:BisimRel_Behaviour}. 
\end{proof}\vspace{0.5cm}

\begin{remark}
  \REFthm{thm:BisimImplExtBehInclusion} does not extend to the asynchronous simulation case, since here we cannot ensure finding pairs $x_1$ and $x'$ s.t. their initial states are related.
\end{remark}

\section{Soundness}
As the main result of this paper we show that the simulation relations in \REFdef{def:SimRel} are well defined by proving that they are preorders for their respective class of state space $\signalmap$-dynamical systems.\\

\begin{theorem}\label{thm:SimAsPreorder}
The relations $\kg{\async}$, $\kg{\wsync}$, $\kg{\sync}$ and $\kg{l}$ are preorders for the class of \textbf{asynchronous} state space $\signalmap$-dynamical systems.
\end{theorem}

\begin{proof}
To simplify notation, we denote the conjunction on the right hand side of \eqref{equ:def:SimRel:b} by $\Omega$, i.e.
\[\Omega(\cdot_a,\cdot_b,\cdot_c):=\begin{propConjA}
\gamma_c=\CONCAT{\gamma_b}{k_c}{k_a}{\gamma_a}\\
\AllQ{t\in\T_c, t< t_c}{
\begin{propConjA}
w_c(t)=w_b(t)\\
x_c(t)=x_b(t)\\
\tau_c(t)=\tau_b(t)\\
\end{propConjA}}\\
x_c(t_c)=x_b(t_c)\\
\AllQSplit{k\geq k_c,t_a'\in\timescaleDown{\tau_a}(k-k_c+k_a),t_a'>t_a}{\ExQ{t_c'\in\timescaleDown{\tau_c}(k),t_c'>t_c }{
\Tuple{x_a(t_a'),x_c(t_c')}\in\R}}
\end{propConjA}.\]
A relation is a preorder, if it is reflexive and transitive.\\
\textbf{1. reflexivity:}\\
To prove reflexivity, pick an arbitrary $\EpS{}=\EpSRhs{}$, construct $\R\subseteq\X\times\X$ s.t. $\propAequ{\Tuple{\xi_1,\xi_2}\in\R}{\xi_1=\xi_2}$ and show that \eqref{equ:def:SimRel}, \eqref{equ:def:SimRelSync}, \eqref{equ:def:SimRelSync:2} and \eqref{equ:def:InitSimRel} hold:
\begin{compactitem}
 \item \eqref{equ:def:SimRel:a}, \eqref{equ:def:SimRelSync:a}, \eqref{equ:def:SimRelSync:2:a} and \eqref{equ:def:InitSimRel:a} hold by construction.
 \item Remember from fact (D) in the proof of \REFlem{lem:LinitSimImpliesOtherSims} that if \eqref{equ:def:SimRel:b} holds, it also holds for $k=k_1=k_2$ and $t=t_1=t_2$.
\item To show that \eqref{equ:def:SimRel:b} holds, fix $\Tuple{w_1,x_1}\in\BehS{},\Tuple{w',x'}\in\BehS{}, \Tuple{\gamma_1,\tau_1}\in\signalmap(w_1),\Tuple{\gamma',\tau'}\in\signalmap(w'), t_1,t_2\in\T, k_1,k_2\in\T_E$ s.t.
the left side of \eqref{equ:def:SimRel:b} is true,
pick  $w_2\in\W^\T$, $x_2\in\X^\T$, $\gamma_2\in\Gamma^{\TE}$, $\tau_2\in\TE^\T$ s.t.
 \begin{align}\label{equ:proof:construct_preoreder}
  \w_2&=\CONCAT{\w'}{t_2}{t_1}{\w_1}&
  \x_2&=\CONCAT{\x'}{t_2}{t_1}{\x_1}&
  \tau_2&=\CONCAT{\tau'}{t_2}{t_1}{(\tau_1+c)}&
  \gamma_2&=\CONCAT{\gamma'}{k_2}{k_1}{\gamma_1}
\end{align}
and show that the right side of \eqref{equ:def:SimRel:b} is true. \\
\begin{inparaitem}[$\blacktriangleright$]
\item Observe that the first three lines of $\Omega(\cdot_1,\cdot',\cdot_2)$ follow directly from \eqref{equ:proof:construct_preoreder} and from the construction of $\R$ implying $x_1(t_1)=x'(t_2)$. \\
\item Now using \REFdef{def:SSDynSysInducesFiniteV} we can conclude that $\Tuple{w_2,x_2}\in\BehS{}$ and $\Tuple{\gamma_2,\tau_2}\in\signalmap(w_2)$ since $\Tuple{w_1,x_1}\in\BehS{},\Tuple{w',x'}\in\BehS{}$ and $x_1(t_1)=x'(t_2)=x_2(t_2)$.\\
\item To show that the last line of $\Omega(\cdot_1,\cdot',\cdot_2)$ is true, observe that \eqref{equ:proof:construct_preoreder} implies $\AllQ{k\geq k_2,t_1'\in\timescaleDown{\tau_1}(k-k_2+k_1),t_1'>t_1,t_2'\in\timescaleDown{\tau_2}(k),t_2'>t_2 }{x_1(t_1')=x_2(t_2')}$. From the construction of $\R$ this implies $\Tuple{x_1(t_1'),x_2(t_2')}\in\R$.
\end{inparaitem}
\end{compactitem}
\textbf{2. transitivity}\\
To prove transitivity, pick arbitrary\footnote{Since the proof is equivalent for all relations, we do not specify them and use $\kg{}$ as their unique representative.} $\EpS{1},\EpS{2},\EpS{3}$ s.t. $\propConj{\BR{\EpS{1}\kg{}\EpS{2}}}{\BR{\EpS{2}\kg{}\EpS{3}}}$. This implies that there exist simulation relations $\R_{1,2}$ and $\R_{2,3}$ from $\EpS{1}$ to $\EpS{2}$ and  $\EpS{2}$ to $\EpS{3}$, respectively. Now construct $\R_{1,3}$ s.t.
\[\propAequ{\Tuple{\xi_1,\xi_3}\in\R_{1,3}}{\ExQ*{\xi_2\in\X_2}{\propConj*{\Tuple{\xi_1,\xi_2}\in\R_{1,2}}{\Tuple{\xi_2,\xi_3}\in\R_{2,3}}}}\] and show that \eqref{equ:def:SimRel}, \eqref{equ:def:SimRelSync}, \eqref{equ:def:SimRelSync:2} and \eqref{equ:def:InitSimRel} hold for $\R_{1,3}$, implying $\EpS{1}\kg{}\EpS{3}$.
\begin{compactitem}
\item Observe that \eqref{equ:def:SimRel:a}, \eqref{equ:def:SimRelSync:a}, \eqref{equ:def:SimRelSync:2:a} and \eqref{equ:def:InitSimRel:a} hold for $\R_{1,2}$ and $\R_{2,3}$, implying
\begin{small}
\begin{align*}
&\AllQ{\xi_1\in\X_1}{\ExQ*{\xi_2\in\X_{2},\xi_3\in\X_{3}}{\propConj*{\Tuple{\xi_1,\xi_2}\in\R_{1,2}}{\Tuple{\xi_2,\xi_3}\in\R_{2,3}},}}\\
&\AllQ{k\in\TE,\xi_1\in\Xk{k}{1}}{\ExQ*{\xi_2\in\Xk{k}{2},\xi_3\in\Xk{k}{3}}{\propConj*{\Tuple{\xi_1,\xi_2}\in\R_{1,2}}{\Tuple{\xi_2,\xi_3}\in\R_{2,3}},}}\\
&\AllQ{t\in\T,\xi_1\in\Xt{t}{1}}{\ExQ*{\xi_2\in\Xt{t}{2},\xi_3\in\Xt{t}{3}}{\propConj*{\Tuple{\xi_1,\xi_2}\in\R_{1,2}}{\Tuple{\xi_2,\xi_3}\in\R_{2,3}},}}\\
&\AllQ{\xi_1\in\Xk{l}{1}}{\ExQ*{\xi_2\in\Xk{l}{2},\xi_3\in\Xk{l}{3}}{\propConj*{\Tuple{\xi_1,\xi_2}\in\R_{1,2}}{\Tuple{\xi_2,\xi_3}\in\R_{2,3}},}}
\end{align*}
\end{small}
respectively. Using the construction of $\R_{1,3}$ this implies that \eqref{equ:def:SimRel:a}, \eqref{equ:def:SimRelSync:a}, \eqref{equ:def:SimRelSync:2:a} and \eqref{equ:def:InitSimRel:a} hold for $\R_{1,3}$.
\item Remember from fact (D) in the proof of \REFlem{lem:LinitSimImpliesOtherSims} that if \eqref{equ:def:SimRel:b} holds, it also holds for $k=k_1=k_2$ and $t=t_1=t_2$.
\item To show
 \eqref{equ:def:SimRel:b}, fix $\Tuple{w_1,x_1}\in\BehS{1},\Tuple{w',x'}\in\BehS{3}, \Tuple{\gamma_1,\tau_1}\in\signalmap_1(w_1),\Tuple{\gamma',\tau'}\in\signalmap_3(w'), t_1\in\T_1, t_3\in\T_3, k_1=\timescaleUp{\tau_1}(t_1), k_3=\timescaleUp{\tau'}(t_3)$ s.t.
$\Tuple{x_1(t_1),x'(t_3)}\in\R_{1,3}$.\\
\begin{inparaitem}[$\blacktriangleright$]
 \item From the construction of $\R_{1,3}$ we know that there exists some 
$\Tuple{w'',x''}\in\BehS{2}$, $\Tuple{\gamma'',\tau''}\in\signalmap_2(w'')$,$t_2\in\T_2$, $k_2=\timescaleUp{\tau_2}(t_2)$ s.t. $\Tuple{x_1(t_1),x''(t_2)}\in\R_{1,2}$ and $\Tuple{x''(t_2),x'(t_3)}\in\R_{2,3}$.\\
\item This implies that we can fix some $\Tuple{w_2,x_2}\in\BehS{2}$, $\Tuple{\gamma_2,\tau_2}\in\signalmap_2(w_2)$ s.t. $\Omega(\cdot_1,\cdot'',\cdot_2)$ holds and therefore  $\Tuple{x_1(t_1),x_2(t_2)}\in\R_{1,2}$ and $\Tuple{x_2(t_2),x'(t_3)}\in\R_{2,3}$.\\ 
\item This implies that we can fix some $\Tuple{w_3,x_3}\in\BehS{3}$, $\Tuple{\gamma_3,\tau_3}\in\signalmap_3(w_3)$ s.t. $\Omega(\cdot_2,\cdot',\cdot_3)$ holds. 
\end{inparaitem}
\item With this choice of signals, we show that $\Omega(\cdot_1,\cdot',\cdot_3)$ also holds:\\
\begin{inparaitem}[$\blacktriangleright$]
\item Observe, that the second and third line of $\Omega(\cdot_1,\cdot',\cdot_3)$ are equivalent to the second and third line of $\Omega(\cdot_2,\cdot',\cdot_3)$, respectively.\\ \item Using the first line of $\Omega(\cdot_1,\cdot'',\cdot_2)$ and $\Omega(\cdot_2,\cdot',\cdot_3) $ we get $\gamma_3=\CONCAT{\gamma'}{k_3}{k_2}{\gamma_2}=\CONCAT{\gamma'}{k_3}{k_2}{\CONCAT{\gamma''}{k_2}{k_1}{\gamma_1}}=\CONCAT{\gamma'}{k_3}{k_1}{\gamma_1}$ implying that the first line of $\Omega(\cdot_1,\cdot',\cdot_3)$ holds. \\
\item Finally, to show that the last line of $\Omega(\cdot_1,\cdot',\cdot_3)$ holds, observe that it is equivalent to 
\begin{equation}\label{equ:proof:2}
 \AllQSplit{k\geq k_2,t_1'\in\timescaleDown{\tau_1}(k-k_2+k_1),t_1'>t_1}{\ExQ{t_2'\in\timescaleDown{\tau_2}(k),t_2'>t_2, t_3'\in\timescaleDown{\tau_3}(k),t_3'>t_3}{
\begin{propConjA}
\Tuple{x_1(t_1'),x_2(t_2')}\in\R_{1,2}\\
\Tuple{x_2(t_2'),x_3(t_3')}\in\R_{2,3}
\end{propConjA}.}}
\end{equation}
To show that \eqref{equ:proof:2} holds, fix $k\geq k_3,t_1'\in\timescaleDown{\tau_1}(k-k_3+k_1),t_1'>t_1,t_3'\in\timescaleDown{\tau_3}(k),t_3'>t_3 $ and pick $t_2'\in\timescaleDown{\tau_2}(k-k_3+k_2),t_2'>t_2$. \\
\begin{inparaitem}[$\triangleright$]
\item With this choice it follows immediately from the last line of $\Omega(\cdot_2,\cdot',\cdot_3)$ that $\Tuple{x_2(t_2'),x_3(t_3')}\in\R_{2,3}$.\\
\item If we now pick $\tilde{k}=k-k_3+k_2$, we have $\tilde{k}\geq k_2$, since $k\geq k_3$.\\
\item Now it follows from $t_1'\in\timescaleDown{\tau_1}(k-k_3+k_1)$ that $t_1'\in\timescaleDown{\tau_1}(\tilde{k}-k_2+k_1)$ and from $t_2'\in\timescaleDown{\tau_2}(k-k_3+k_2)$ that $t_2'\in\timescaleDown{\tau_2}(\tilde{k})$. \\
\item Using the last line of $\Omega(\cdot_1,\cdot'',\cdot_2)$ this implies that $\Tuple{x_1(t_1'),x_2(t_2')}\in\R_{1,2}$%
\end{inparaitem}%
\end{inparaitem}.
\end{compactitem}\vspace{-0.4cm}
\end{proof}\vspace{0.5cm}

\begin{theorem}\label{thm:SimAsPreorder:wsync}
The relations $\kg{\wsync}$ and $\kg{\sync}$ are preorders for the class of \textbf{externally synchronous} state space $\signalmap$-dynamical systems. 
\end{theorem}

\begin{proof}
 This proof is identical to the proof of \REFthm{thm:SimAsPreorder} by using $k=k_1=k_2$ in all statements. This substitution is applicable since \eqref{equ:def:SimRel:b} is also restricted to $k=k_1=k_2$ for $\kg{\wsync}$ and $\kg{\sync}$.
\end{proof}\vspace{0.5cm}

\begin{theorem}\label{thm:SimAsPreorder:sync}
The relation $\kg{\sync}$ is a preorder for the class of \textbf{synchronous} state space $\signalmap$-dynamical systems.
\end{theorem}

\begin{proof}
This proof is identical to the proof in \REFthm{thm:SimAsPreorder} by using $k=k_1=k_2$ and  $t=t_1=t_2$ in all statements. This substitution is applicable since \eqref{equ:def:SimRel:b} is also restricted to $k=k_1=k_2$ and  $t=t_1=t_2$ for $\kg{\sync}$.
\end{proof}\vspace{0.5cm}

\begin{corollary}\label{thm:BisimAsEquivalenceRelation}
The relations $\hg{\async}$, $\hg{\wsync}$, $\hg{\sync}$ and $\hg{l}$ are equivalence relations for the class of \textbf{asynchronous} state space $\signalmap$-dynamical systems. 
Furthermore, the relations $\hg{\wsync}$ and $\hg{\sync}$ are equivalence relations for the class of \textbf{externally synchronous} state space $\signalmap$-dynamical systems, and the relation $\hg{\sync}$ is an equivalence relation for the class of \textbf{synchronous} state space $\signalmap$-dynamical systems.
\end{corollary}

\begin{proof}
A relation is an equivalence relation, if it is reflexive, transitive and symmetric. 
From \REFdef{def:BisimRel_Behaviour}, it follows that all relations $\hg{}$ are defined by two simulation relations. Therefore reflexivity and transitivity follows from \REFthm{thm:SimAsPreorder}~-~\ref{thm:SimAsPreorder:sync}.

To prove symmetry, pick arbitrary $\EpS{1},\EpS{2}$ and show $\propImp{\BR{\EpS{1}\hg{}\EpS{2}}}{\BR{\EpS{2}\hg{}\EpS{1}}}$. Observe that it follows immediately from \REFdef{def:BisimRel_Behaviour} that for any bisimulation relation $\R$ between $\EpS{1}$ and $\EpS{2}$ we can pick $\tilde{R}=\R^{-1}$ as a bisimulation relation between $\EpS{2}$ and $\EpS{1}$, implying $\EpS{2}\hg{}\EpS{1}$.
\end{proof}\vspace{0.5cm}

\section{Conclusion}
 We have proposed a behavioral system model with distinct external and internal signals possibly evolving on different time scales. For this new system model different notions of simulation and bisimulation were derived and their soundness was proven. In Remarks~\ref{rem:0},~\ref{rem:GFS} and \ref{rem:TS}, we discussed in an intuitive manner that our notion can capture a broad selection of similarity concepts available in the literature. The formal proofs of these intuitive connections will be presented in a subsequent paper. 
 It is our goal for the near future to use the presented framework to compare existing abstraction techniques in the control systems community.

\end{document}